%% file: BCT_compression-TPrev.tex
\documentclass[a4paper,superscriptaddress,english,accepted=2024-04-03]{quantumarticle}
\pdfoutput=1
\usepackage[T1]{fontenc}
\usepackage[utf8]{inputenc}
\usepackage{lipsum,babel}
\usepackage{mathtools}
\usepackage{amsthm}
\usepackage{hyperref}
\usepackage{enumitem}

\usepackage{amsmath,amssymb,mathrsfs}

\usepackage{graphicx}
\usepackage{tikzit}
\input{sample.tikzstyles}
\input{myqcircuit.tex}

\newcommand{\transf}[1]{\ensuremath{\mathscr{#1}}}
\newcommand{\tA}{\transf A}
\newcommand{\tB}{\transf B}
\newcommand{\tC}{\transf C}
\newcommand{\tD}{\transf D}
\newcommand{\tE}{\transf E}
\newcommand{\tF}{\transf F}
\newcommand{\tG}{\transf G}

\newcommand{\tR}{\transf R}
\newcommand{\tS}{\transf S}
\newcommand{\tT}{\transf T}
\newcommand{\tU}{\transf U}
\newcommand{\tV}{\transf V}

\newcommand{\tI}{\transf I}
\newcommand{\sys}[1]{\ensuremath{\mathrm{#1}}}
\newcommand{\rA}{\sys A}
\newcommand{\rB}{\sys B}
\newcommand{\rC}{\sys C}
\newcommand{\rD}{\sys D}
\newcommand{\rE}{\sys E}
\newcommand{\rF}{\sys F}

\newcommand{\rI}{\sys I}

\newcommand{\rX}{\sys X}

\newcommand{\st}{\mathsf{St}}
\newcommand{\purst}{\mathsf{PurSt}}

\newcommand{\eff}{\mathsf{Eff}}
\newcommand{\tr}{\mathsf{Tr}}

\newcommand{\R}{\mathbb{R}}
\newcommand{\N}{\mathbb{N}}

\newcommand{\bi}{\mathbf{i}}

\newtheorem{theorem}{Theorem}[section]
\newtheorem{definition}{Definition}[section]
\newtheorem{proposition}{Proposition}[section]
\newtheorem{lemma}{Lemma}[section]

\newtheorem{property}{Property}

\DeclarePairedDelimiter{\norma}{\lVert}{\rVert}
\DeclarePairedDelimiter{\normop}{\lVert}{\rVert_{\mathrm {op}}}

\DeclarePairedDelimiter{\pairing}{(}{)}
\DeclarePairedDelimiter{\rket}{|}{)}
\DeclarePairedDelimiter{\rbra}{(}{|}

\newcommand{\cP}{\mathcal{P}}

\newcommand{\cO}{\mathcal{O}}
\newcommand{\cR}{\mathcal{R}}

\newcommand{\sX}{\mathsf{X}}
\newcommand{\sY}{\mathsf{Y}}

\newcommand{\sL}{\mathsf{L}}

\newcommand{\is}{\mathsf{\bf{i_s}}}
\newcommand{\p}{\mathsf{\bf{p}}}
\newcommand{\q}{\mathsf{\bf{q}}}
\newcommand{\s}{\mathsf{\bf{s}}}
\newcommand{\sti}{\mathsf{\bf{i}}}
\newcommand{\ts}{\mathsf{\bf{t_{n}}}}
\newcommand{\tnbar}{\mathsf{\bf{\bar{t}_{\bar{n}}}}}
\newcommand{\js}{\mathsf{\bf{j_s}}}
\newcommand{\hf}{\mathsf{\bf{h}}}

\usepackage{xcolor}

\begin{document}

\title{	Which entropy for general physical theories?}

\author{Paolo Perinotti}

\email{paolo.perinotti@unipv.it}

\affiliation{QUIT group, Physics Dept., Pavia University, and INFN Sezione di Pavia, via Bassi 6, 27100 Pavia, Italy}

\author{Alessandro Tosini}

\email{alessandro.tosini@unipv.it}

\affiliation{QUIT group, Physics Dept., Pavia University, and INFN Sezione di Pavia, via Bassi 6, 27100 Pavia, Italy}

\author{Leonardo Vaglini}

\email{leonardo.vaglini01@universitadipavia.it}

\affiliation{QUIT group, Physics Dept., Pavia University, and INFN Sezione di Pavia, via Bassi 6, 27100 Pavia, Italy}

\begin{abstract}
We address the problem of quantifying the information content of a source for an arbitrary information theory, where 
the information content is defined in terms of the asymptotic achievable compression rate. The functions that solve 
this problem in classical and quantum theory are Shannon's and von Neumann's entropy, respectively. However, in a 
general information theory there are three different functions that extend the notion of entropy, and this opens the 
question as to whether any of them can universally play the role of the quantifier for the information content. Here 
we answer the question in the negative, by evaluating the information content as well as the various entropic 
functions in a toy theory called Bilocal Classical Theory.
\end{abstract}
\maketitle

\section{Introduction}

The experience of quantum information has suggested that viewing quantum theory as an information theory helps in making its  most counterintuitive consequences comprehensible. This perspective 
promotes the role of information to a fundamental one in the development of a physical 
theory~\cite{hardy1999disentangling,Spekkens:2007aa,hardy2001quantum,PhysRevA.84.012311,DAriano:2017aa}. For example, one might 
wonder to what extent the laws of physics depend on the properties of physical systems used as information carriers.
Such an approach to foundations of physics already proved successful both in motivating the use of quantum theory on 
operational grounds, and in suggesting how viable generalisations of quantum theory for a post-quantum physics should 
look 
like~\cite{PhysRevA.75.032304,Pawowski:2009aa,PhysRevA.81.062348,hayden2007black,muller2012black,Pastawski:2015aa}.

From the above perspective, it is very natural to introduce in a general theory of physical systems some notion of 
entropy at a very fundamental level. Indeed, entropic functions have been identified as quantifiers of information 
from a source since the pioneering work by Shannon that paved the way to classical information 
theory~\cite{shannon1949communication}. Along this line, von Neumann entropy~\cite{VonNeumann:1927ve} was studied in 
the realm of quantum information theory as a quantifier of {\em quantum information}~\cite{PhysRevA.51.2738}.

The importance of Shannon's and von Neumann's entropy is due to classical and quantum source coding theorems, proving
that these two functions represent the achievable compression rates for a classical and a quantum source, 
respectively~\cite{shannon1949communication,PhysRevA.51.2738}.

In order to generalise the quantification of information content to a wider scenario of possible information 
theories---i.e.~theories of systems and their processings~\cite{PhysRevA.75.032304,PhysRevA.81.062348}---various 
authors started from the alternate definitions of Shannon's and von Neumann's entropies that refer to state 
preparations, measurements and their use in a communication scenario, and, as such, they have a counterpart in a 
general information theory~\cite{Barnum:2010gy,Short:2010kt,KIMURA2010175}. 

The first problem that one has to face at this point is that the possible definitions are three, and in general they 
lead to three functions that are actually different in the general case~\cite{KIMURA2010175}. A second, even more 
compelling problem, is that it is unknown whether any of the three entropic functions mentioned above is an exact 
quantifier corresponding to the best achievable compression rate of a source. In the light of these questions, 
classical and quantum theory are very special and fortunate examples of theories where both problems have a known 
solution: both theories are indeed {\em monoentropic}---i.e.~ the three entropic functions 
coincide~\cite{Barnum:2010gy}---and the unique entropy exactly matches the information content of a source, 
defined as its best achievable compression rate.

In this scenario, one can then ask two crucial and independent questions: i) 
what are the features that make a theory monoentropic; ii) under what 
conditions is the information content quantified by at least one of the three 
possible entropies---or a regularized version of them.


In order to tackle the second question, in Ref.~\cite{PhysRevA.105.052222} the present authors set a minimal framework 
where the question can be addressed, and they introduced a thorough definition of information content, proving a 
lower bound for such quantity in terms of one of the three entropies. The latter result was then used to prove that a 
mixed state cannot have null information content, while, on the other hand, under suitable conditions, a pure state 
has null information content. In turn, the latter result opened a further question: do there exist theories allowing for pure states that have non-null information content?

In the present paper, we study the above questions in the theory called Bilocal Classical Theory 
(BCT)~\cite{PhysRevA.102.052216}, and show that, while BCT is monoentropic, the information content of its states {\em 
does not coincide} with its entropy. Moreover, in BCT pure states actually have non-null information content. The 
latter feature can be understood considering that the independent preparation of two systems in pure 
states does not correspond to a pure state for the composite system. Such a preparation thus introduces some ignorance 
about the whole system, even if the outcomes of independent experiments on its components are fully predictable. The 
latter result can be extended to any theory of classical systems where the rule for composing systems is such that the 
composition of pure states is not necessarily pure. Therefore, thinking of pure states as representing complete 
knowledge of the physical system is inaccurate if the composition law is not purity-preserving.

\section{An overview of the framework}

A brief review of the framework follows. A standard reference is \cite{DAriano:2017aa}, but an extended treatment can be found in \cite{PhysRevA.81.062348}. The framework is also carefully reviewed in the introduction of \cite{PhysRevA.102.052216}
\subsection{Review of the framework}\label{subsec:framework}

The primitive notions of an operational theory are those of test, event and system. A \emph{test} $\{\tA_i\}_{i\in\sX}$
is given by a collection of \emph{events}, where $i$ labels the elements
 of the {\em outcome space} $\sX$. The \emph{systems} (or, more precisely, {\em system types})
allow for the connection between different tests, and are denoted by capital Roman letters $\rA,\rB,\dots$. Therefore, 
a test is completely determined by its input and output systems, and the events associated with the outcome space $\sX$.
In order to represent a test and its events $\{\tA_i\}_{i\in\sX}$ we use the usual diagrammatic notation
\begin{align*}
  \begin{aligned}
    \Qcircuit @C=1em @R=.7em @! R {&\ustick{\rA} \qw&\gate{\{\tA_i\}_{i\in\sX}}& \ustick{\rB}
      \qw&\qw}
  \end{aligned}
  \quad , \quad
  \begin{aligned}
    \Qcircuit @C=1em @R=.7em @! R {&\ustick{\rA}\qw&\gate{\tA_i}&\ustick{\rB} \qw &\qw}
  \end{aligned}\ ,\\
\end{align*}
and we will call $\rA$, $\rB$ the input and the output system of the test, respectively. In this case we will say that 
the test $\{\tA_i\}_{i\in\sX}$, as well as each of its events $\tA_i$, are of type $\rA\to\rB$.  
If $\{\tA_i\}_{i\in\sX}$ and $\{\tB_j\}_{j\in\sY}$ are two tests,
one can define their \emph{sequential composition} as the test $\{\tC_{i,j}\}_{(i,j)\in\sX\times\sY}$, with events 
$\tC_{i,j}$ that are diagrammaticaly represented by 
\begin{align*}
  \begin{aligned}
    \Qcircuit @C=1em @R=.7em @! R {&\ustick{\rA} \qw&\gate{\tC_{i,j}}& \ustick{\rC}
      \qw&\qw}
  \end{aligned}
  \quad&\coloneqq\quad
  \begin{aligned}
    \Qcircuit @C=1em @R=.7em @! R {&\ustick{\rA}\qw&\gate{\tA_i}&\ustick{\rB} \qw &\gate{\tB_{j}}&\ustick{\rC} \qw
      &\qw}
  \end{aligned}\ .
\end{align*}
Notice that this definition requires the output system of the events on the left to be necessarily the input system of 
the events on the right. In formula, the sequential composition will be denote as $\tC_{i,j}\coloneqq\tB_j\tA_i$. A \emph{singleton test} is a test whose outcome space set $\sX$ is a singleton, and the 
unique event contained in it is called \emph{deterministic}. For any system $\rA$ there exists a unique 
\emph{identity} test $\{\tI_{\rA}\}$ such that, for any event of type $\rA\to\rB$, one has $\tA\tI_{\rA}=\tI_{\rB}\tA=\tA$. Another operation that can be performed on tests for defining a new test is \emph{parallel 
composition}. Given two systems $\rA$ and $\rB$ we call $\rA\rB$ the composite system of $\rA$ and $\rB$. Then, if 
$\{\tA_i\}_{i\in\sX}$ and $\{\tC_j\}_{j\in\sY}$ are two tests of type $\rA\to\rB$ and $\rC\to\rD$, respectively, their 
parallel composition is the test $\{\tA_i\boxtimes\tC_j\}_{(i,j)\in\sX\times\sY}$. Diagrammatically
\begin{align*}
  \begin{aligned}
    \Qcircuit @C=1em @R=1em @! R {&\ustick{\rA} \qw&\multigate{1}{\tA_i\boxtimes\tC_j}& \ustick{\rB} \qw&\qw \\
							&\ustick{\rC} \qw&\ghost{\tA_i\boxtimes\tC_j}& \ustick{\rD} \qw&\qw }
  \end{aligned}
  \quad&\coloneqq\quad
  \begin{aligned}
    \Qcircuit @C=1em @R=1em @! R {&\ustick{\rA} \qw&\gate{\tA_i}&\ustick{\rB} \qw & \qw \\
							&\ustick{\rC} \qw&\gate{\tC_{j}}&\ustick{\rD} \qw&\qw}
  \end{aligned} \, .
\end{align*}
Both the parallel and sequential compositions must be associative,
and the parallel composition operation commutes with the sequential one, namely 
$(\tE_h\boxtimes \tF_k)(\tC_i\boxtimes \tD_j)=(\tE_h \tC_i)\boxtimes (\tF_k \tD_j)$.
The identity $\tI_{\rA\rB}$ of the composite system $\rA\rB$ is the parallel composition of the two identities $\tI_\rA\boxtimes\tI_\rB$
\begin{equation}\label{eq:CompId}
 \begin{aligned}
    \Qcircuit @C=1em @R=.7em @! R {&\qw & \ustick {\rA\rB}\qw &\qw &\qw}
  \end{aligned}
  \quad \coloneqq\quad
  \begin{aligned}
    \Qcircuit @C=1em @R=1.7em @! R {&\qw &\ustick{\rA}\qw &\qw &\qw \\
	&\qw &\ustick{\rB} \qw &\qw &\qw}
  \end{aligned}\quad.
\end{equation}
We will denote by $\rA^{\boxtimes k}$ the composition
of the same system $\rA$ with itself $k$ times.

There is a special kind of system, the \emph{trivial system} I, satisfying $\rA\rI=\rI\rA=\rA$ for every system $\rA$. 
Tests with $\rI$ as input system and $\rA$ as the output one are called \emph{preparation tests} of $\rA$, while tests 
with input system $\rA$ and $\rI$ as output are named \emph{observation tests} of $\rA$. The events of a preparation 
test $\{\rho_i\}_{i\in\sX}$ and of an observation test $\{a_j\}_{j\in\sY}$ are represented through the following 
diagrams
\begin{align*}
  \begin{aligned}
    \Qcircuit @C=1em @R=.7em @! R {\prepareC{\rho_i}& \ustick \rA
      \qw&\qw}
  \end{aligned}
  \quad&\coloneqq\quad
  \begin{aligned}
    \Qcircuit @C=1em @R=.7em @! R {&\ustick{\rI}\qw&\gate{\rho_i}&
      \ustick \rA \qw&\qw}
  \end{aligned}\quad,\\
  \\
    \begin{aligned}
    \Qcircuit @C=1em @R=.7em @! R {& \ustick \rA\qw &\measureD {a_j}}
  \end{aligned}
  \quad&\coloneqq\quad
  \begin{aligned}
    \Qcircuit @C=1em @R=.7em @! R {&\ustick{\rA}\qw&\gate{a_j}&
      \ustick \rI \qw&\qw}
  \end{aligned}\quad.
\end{align*}
In the following we will always use Greek letters to denote preparation tests and Latin letters for the observation 
ones. Preparation and observation events will also be denoted by using round brackets: $\rket{\rho}_\rA$ and 
$\rbra{a}_{\rA}$, respectively, and the system will be omitted whenever it is clear from the context.

A \emph{circuit} is a diagram representing an arbitrary test that is obtained by sequential and parallel composition of other tests. We say that a circuit is \emph{closed} when the input and output systems are both the trivial one, namely, when it starts with a preparation test 
and it ends with an observation test. For any pair of systems, we want their agents to be allowed to exchange their systems. This requirement is captured by the notion of \emph{braiding}, that is
a family of reversible transformations $\tS_{\rA\rB}$, defined for any pair of systems of the theory, and denoted as follows
\begin{align}
  \begin{aligned}
    \Qcircuit @C=1em @R=1em @! R {&\ustick{\rA} \qw&\multigate{1}{\tS_{\rA\rB}}& \ustick{\rB} \qw&\qw \\
							&\ustick{\rB} \qw&\ghost{\tS_{\rA\rB}}& \ustick{\rA} \qw&\qw }
  \end{aligned}
& \ = \ 
\input{braiding.tikz}\ , \\[10pt]
  \begin{aligned}
    \Qcircuit @C=1em @R=1em @! R {&\ustick{\rA} \qw&\multigate{1}{\tS^{-1}_{\rA\rB}}& \ustick{\rB} \qw&\qw \\
							&\ustick{\rB} \qw&\ghost{\tS_{\rA\rB}}& \ustick{\rA} \qw&\qw }
  \end{aligned}
& \ = \  
\input{braiding_inv.tikz}\ .
\end{align}
Generally, $\tS_{\rA\rB}$ and $\tS_{\rA\rB}^{-1}$ are different transformations. When they are equal, e.g. in the case of QT, where $\tS_{\rA\rB}$ is represented by the swap operator,
the theory is called symmetric. These transformations must obey a sliding property, which asserts that two agents can equivalently perform their transformations and then exchange their
output systems, or first exchange their input systems and then perform the transformations. In diagrams
\begin{equation}
\input{swap_LHS.tikz}\ = \ \input{swap_RHS.tikz}\ .
\end{equation}
Then,
an \emph{operational Theory} is defined by a collection of systems, that is closed under composition, and a collection of tests (including all the family of 
tests $\tS_{\rA\rB}$ for any $\rA$ and $\rB$) closed under a pair of associative operations of sequential and parallel composition.

An \emph{operational 
probabilistic theory} is an operational theory where the test corresponding to any closed
circuit (equivalently, any test from the trivial system to itself) is given by a joint probability distribution 
for its outcomes, conditioned by the tests that make the circuit. Moreover, compound tests from the trivial system to 
istelf are independent, namely, the joint probability distribution is assumed to be given by the product of the probability 
distributions of the composing tests. A simple example is given by a preparation test $\{\rho_i\}_{i\in\sX}$ 
sequentially followed by an observation test $\{a_j\}_{j\in\sY}$:
\begin{align*}
\{p(i,j|\{\rho_l\},\{a_m\})\}\coloneqq
\left\{
 \begin{aligned}
    \Qcircuit @C=1em @R=.7em @! R {\prepareC{\rho_i}& \ustick \rA
      \qw&\measureD{a_j}}
  \end{aligned}
\right\}  ,
\end{align*}
 with $\sum_{i,j}p(i,j)=1$. Thus, one has a joint probability distribution, which is conditioned by the chosen tests 
$\{\rho_i\}_{i\in\sX}$ and  $\{a_j\}_{j\in\sY}$. From now on we will simply omit this dependence. 
The probability associated with the closed circuit where a preparation $\rho_i$ is followed by an observation $a_j$ 
will also be denoted by a pairing, $p(i,j)=(a_j|\rho_i)$.

Given any system $\rA$ of an OPT, One can define an equivalence relation on the set of preparation events by declaring 
that $\rho\sim\sigma$ iff  $\pairing{a |\rho}= \pairing{a |\sigma}$ for any observation event $a$. The set of 
equivalence classes with respect to this relation is called the set of \emph{states} of system $\rA$ and it is denoted 
by $\st(\rA)$. Similarly, one can define the set of \emph{effects} as the set of equivalence classes of the 
observation events such that $\pairing{a|\rho}=\pairing{b|\rho}$ for any preparation event $\rho$, and this is 
denoted by $\eff(\rA)$. The sets of deterministic states and effects, obtained as the equivalence classes of deterministic preparation events and deterministic observation events respectively, will be denoted by $\st_1({\rA})$ and 
$\eff_1({\rA})$ respectively. By definition, effects are separating for states, namely $\rket{\rho}\neq\rket{\sigma}$, 
iff there exists an effect $\rbra{a}$ such that $(a|\rho)\neq(a|\sigma)$, and viceversa states are separating for 
effects.

Given the probabilistic structure, states can be seen as functionals on the set of effects and viceversa, and then one 
can consider linear combinations of them, thus defining two linear spaces, $\st(\rA)_{\R}$ and $\eff(\rA)_{\R}$, which 
are dual to each other---assuming that they are finite-dimensional. The \emph{size} $D_{\rA}$ of a given system $\rA$ 
is defined as the dimension of the linear space $\st_{\R}(\rA)$.
Also the set of events of generic type $\rA\to\rB$ can be endowed with an equivalence relation. Indeed, given $\tA$ and $\tB$ of the 
same type $\rA\to\rB$, we say that they are \emph{operationally equivalent} ($\tA\sim\tB$) if the following identity 
holds
\begin{align*}
\begin{aligned}
\Qcircuit @C=0.8em @R=1.0em @! R { \multiprepareC{1}{\Psi} &\ustick \rA \qw & \gate {\tA} &\ustick  \rB \qw &\multimeasureD{1}{A}  \\
\pureghost {\Psi} &\qw & \ustick \rC  \qw   &\qw  &\ghost{A} }
\end{aligned}
\ = \
\begin{aligned}
\Qcircuit @C=0.8em @R=1.0em @! R { \multiprepareC{1}{\Psi} &\ustick \rA \qw & \gate {\tB} &\ustick  \rB \qw &\multimeasureD{1}{A}  \\
\pureghost {\Psi} &\qw & \ustick \rC  \qw   &\qw  &\ghost{A} }
\end{aligned} \ ,
\end{align*}
for any preparation event $\Psi$, observation event $A$ and any ancillary system $C$. We then denote the set of all the equivalence classes with $\tr(\rA\to\rB)$, whose elements are simply called \emph{transformations}. The operations of sequential and parallel composition are then extended to the equivalence classes in the obvious way, by means of the representatives. For instance, if $\tA$ and $\tB$ are two transformation events of type $\rA\to\rB$ and $\rB\to\rC$ respectively, then the sequential composition between the equivalence classes is defined as $[\tA][\tB]\coloneqq[\tA\tB]$. From now on we will always refer to transformations omitting the square bracket in the notation. Via the sequential composition, a transformation of type $\rA\to\rB$ is associated with a map from $\st(\rA\rC)$ to $\st(\rB\rC)$ for any ancillary system $\rC$, which can be uniquely extended to linear maps between the vector spaces $\st_{\R}(\rA\rC)$ to $\st_{\R}(\rB\rC)$ (see \cite{DAriano:2017aa}). Therefore, a transformations is completely characterised by a family of linear maps, one for each ancillary system $\rC$.
As for states and effects, the set of deterministic transformations will be denoted by $\tr_1(\rA\to\rB)$.
If $\tU\in\tr(\rA\to\rB)$ and there exists $\tV\in\tr(\rB\to\rA)$ such that $\tV\tU=\tI_{\rA}$ and $\tU\tV=\tI_{\rB}$, we say that $\tU$ is \emph{reversible}.
Accordingly, two systems $\rA$ and $\rB$ are called {\em operationally equivalent} ($\rA\cong\rB$) if there exists a reversible 
transformation $\tU\in\tr_1(\rA\to\rB)$. A notion that will be useful in the following is that of {\em asymptotically 
equivalent systems}.

\begin{definition}[Asymptotic equivalence]\label{def:asympeq}
Two systems $\rA_1$ and $\rA_2$ are asymptotically equivalent if the following conditions are satisfied
\begin{enumerate}
\item\label{it:kappa1}
there exists a pair of integers $k_1,k_2<\infty$, $\tE\in\tr_1(\rA_1^{\boxtimes k_1}\to\rA_2^{\boxtimes k_2})$ and $\tD\in\tr_1(\rA_2^{\boxtimes k_2}\to\rA_1^{\boxtimes k_1})$ such that $\tD\tE=\tI_{\rA_1^{\boxtimes k_1}}$;
\item\label{it:hacca2}
there exists a pair of integers $h_1,h_2<\infty$, $\tG\in\tr_1(\rA_2^{\boxtimes h_2}\to\rA_1^{\boxtimes h_1})$ and $\tF\in\tr_1(\rA_1^{\boxtimes h_1}\to\rA_2^{\boxtimes h_2})$ such that $\tF\tG=\tI_{\rA_2^{\boxtimes h_2}}$;
\item\label{it:equiciccia} let $M_{2}^{\rm min}(k_1)$ be the smallest $k_2$ such that item~\ref{it:kappa1} is satisfied for a given $k_1$, and similarly for 
$M_1^{\rm min}(h_2)$ with reference to item~\ref{it:hacca2}. The following relations hold:
\begin{equation}
\lim_{k_1\rightarrow \infty}\frac{M_2^{\rm min}(k_1)}{k_1}=k, \quad \lim_{h_2\rightarrow \infty}\frac{M_1^{\rm min}(h_2)}{h_2}=k^{-1}.
\label{eq:digitizability2}
\end{equation}
\end{enumerate}
\end{definition}
Notice that item \ref{it:kappa1} in the above definition is equivalent to the following: for any $k_1<\infty$ there exists $k_2<\infty$ and a pair of maps $\tE\in\tr_1(\rA_1^{\boxtimes k_1}\to\rA_2^{\boxtimes k_2})$ and $\tD\in\tr_1(\rA_2^{\boxtimes k_2}\to\rA_1^{\boxtimes k_1})$ such that $\tD\tE=\tI_{\rA_1^{\boxtimes k_1}}$.
Indeed, the above statement implies item \ref{it:kappa1}, as we now prove. 

Assuming the validity of item 1 in definition 2.1, if $k=k_1$, the statement is trivially true. We then first prove that for every integer $k<k_1$ there exists another integer $k_2$ and a pair of maps $\tilde{\tE}$ and $\tilde{\tD}$ such that $\tilde{\tD}\tilde{\tE}=\tI_{\rA_1^{\boxtimes k}}$. Let $\tE\in\tr_1(\rA_1^{\boxtimes k_1}\to\rA_2^{\boxtimes k_2})$ and $\tD\in\tr_1(\rA_2^{\boxtimes k_2}\to\rA_1^{\boxtimes k_1})$ be the maps given by item 1, such that 
$\tD\tE=\tI_{\rA_1^{\boxtimes k_1}}$. Let $|\rho)\in\st(\rA_1^{\boxtimes (k_1-k)})$ and define the map $\tE_\rho$ as the sequential composition of the encoding map $\tE$, preceeded by the parallel composition of $|\rho)$---which is a map from the trivial system $\rI$ to $\rA_1^{\boxtimes (k_1-k)}$---with the identity $\tI_{\rA^{\boxtimes k}}$ on the remaining $k$ systems, i.e., 
\[
\tE_\rho\coloneqq\tE\circ(|\rho)\boxtimes \tI_{\rA_1^{\boxtimes k}})\in\tr_1(\rA_1^{\boxtimes k}\to\rA_2^{\boxtimes k_2}).
\] 
Similarly we do for the decoding: we take the decoding $\tD$ and a deterministic effect $(e|\in\eff_1(\rA^{\boxtimes (k_1-k)})$ (i.e. a map from $\rA_1^{\boxtimes (k_1-k)}$ to the trivial system) and we define the sequential composition of $\tD$ with $(e|\boxtimes \tI_{\rA_1^{\boxtimes k}}$, to obtain
\[
\tD_e\coloneqq((e|\boxtimes \tI_{\rA_1^{\boxtimes k}})\circ\tD\in\tr_1(\rA_2^{\boxtimes k_2}\to\rA_1^{\boxtimes k}). 
\]
It is now clear that $\tD_e\tE_\rho=\tI_{\rA_1^{\boxtimes k}}$, indeed, using the fact that parallel and sequential compositions commute
\[
\begin{aligned}
\tD_e\tE_\rho&=[((e|\boxtimes \tI_{\rA_1^{\boxtimes k}})\circ\tD]\circ [\tE\circ(|\rho)\boxtimes \tI_{\rA_1^{\boxtimes k}})] \\
			&=((e|\boxtimes \tI_{\rA_1^{\boxtimes k}}) (|\rho)\boxtimes \tI_{\rA_1^{\boxtimes k}}) \\
			&=(e|\rho)_{\rA_1^{\boxtimes (k_1-k)}}\tI_{\rA_1^{\boxtimes k}}=\tI_{\rA_1^{\boxtimes k}}
\end{aligned}
\]
where the last step follows from the fact that we have chosen a deterministic state $|\rho)$ and a deterministic effect $(e|$ to define the encoding $\tE_\rho$ and the decoding $\tD_e$ respectively, so that $(e|\rho)=1$. Now, for $k>k_1$, let $m$ be the integer such that $(m-1)k_1\leq k< mk_1$ and let $\tilde{k}\coloneqq mk_1-k\leq k_1$. We prove that there exists an integer $\tilde{k}_2$ and a pair $\tilde{\tE}\in\tr_1(\rA_1^{\boxtimes k}\to\rA_2^{\boxtimes k_2})$ and $\tilde{\tD}$ that perfectly encodes $\rA_1^{\boxtimes k_1}$ onto $\rA_2^{\boxtimes k_2}$, i.e., $\tilde{\tD}\tilde{\tE}=\tI_{\rA_1^{\boxtimes k_1}}$. For $\tilde{k}\leq k_1$ we have seen that there exists a pair $(\tilde{\tE},\tilde{\tD})$ such that $\tilde{\tD}\tilde{\tE}=\tI_{\rA_1^{\boxtimes (k_1-\tilde{k})}}$. Then consider the encoding map from $\rA_1^{\boxtimes k}$ to $\rA_2^{\boxtimes mk_2}$ defined as the composition of $\tilde{\tE}$ with $m-1$ copies of $\tE$, say $\tilde{\tE}\boxtimes\tE^{\boxtimes (m-1)k_1}$. Similarly for the decoding, take $\tilde{\tD}\boxtimes\tD^{\boxtimes(m-1)k_1}$. It is now clear that this pair gives a perfect encoding of $\rA_1^{\boxtimes k}$ onto $\rA_2^{mk_2}$. This ensures that the statement in item 3 of definition \ref{def:asympeq} definitely makes sense.

Now we set up some terminology and we introduce pure and mixed states, as well as the definition of state dilation.
\begin{definition}[Refinement and convex refinement of an event]
Let $\tC\in\tr(\rA\to\rB)$.
\begin{enumerate}
\item  A {\em refinement} of $\tC$ is given by a 
	collection of events $\{\tB_{j}\}_{j\in\sY}\subseteq\tr(\rA\to\rB)$
	such that there exists a test $\{\tB_{i}\}_{i\in\sX}$ with $\sY\subseteq\sX$ and
	$\tC=\sum_{j\in\sY}\tB_{j}$.
	We say that a refinement $\{\tD_{i}\}_{i\in\sY}$ is trivial if $\tD_{i}=\lambda_{i}\tC$, $\lambda_{i}\in[0,1]$
	for every $i\in\sY$. Conversely, $\tC$ is called the coarse-graining of the events $\{\tD_{i}\}_{i\in\sY}$. We denote by $\cR(\tC)$ the collection of
all the refinements of $\tC$.
\item A {\em convex refinement} or {\em decomposition} of $\tC$ is given by a collection of pairs $\{(p_j,\tB_{j})\}_{j\in\sY}\subseteq\mathbb R\times\tr(\rA\to\rB)$ where $\{p_j\}_{j\in\sY}$ is a probability distribution and $\tB_j$ are events,
	such that there exist tests $\{\tB^{(j)}_{i}\}_{i\in\sX}$, with $\tB^{(j)}_{i_0}=\tB_j$ for all $j\in\sY$, for which the collection $\{p_j\tB^{(j)}_i\}_{(i,j)\in\sX\times\sY}$ is a legitimate conditional test, and $\tC=\sum_{j\in\sY}p_j\tB_{j}$. We say that a convex refinement is trivial if $\tB_{j}=\tC$ for any $j\in\sY$.
\end{enumerate}
\end{definition}

\begin{definition}[Atomic, refinable and extremal events]
An event $\tC$ is called {\em atomic} if it admits only trivial refinements, and {\em refinable} if it is not atomic. $\tC$ is called {\em extremal} if it admits
only trivial convex refinements.
\end{definition}
The notion of refinement and refinable events give rise to the definitions of \emph{pure} and \emph{mixed} states. 
\begin{definition}[Pure and mixed states]
$\rho\in\st(\rA)$ is called {\em pure} if it is extremal and deterministic. We will denote by $\purst(\rA)$ the set of all the pure 
states of system $\rA$.  On the other hand, $\rho\in\st(\rA)$ is called {\em mixed} if it is 
neither atomic nor extremal. 
\end{definition}
Given any state $\rho\in\st_1(\rA)$, any convex decomposition made of pure states will be accordingly called pure. The set of pure convex decompositions
of $\rho$ will be denoted by $\cP(\rho)=\{\{(p_i,\phi_i)\}_{i\in\sX}| \sum_ip_i\phi_i=\rho,\; \{\phi_i\}_{i\in\sX}\subseteq\purst(\rA),\;\sum_ip_i=1\}$.
\begin{definition}\label{def:dilation}
Let $\rho\in\st(\rA)$ and $\Psi\in\st(\rA\rB)$. We say that $\Psi$ is a {\em dilation} of $\rho$ if there exists a deterministic effect $e\in\eff_1(\rB)$ such that
\begin{equation*}
\begin{aligned}
 \Qcircuit @C=1em @R=.7em @! R { \prepareC{\rho} & \ustick \rA \qw & \qw}
\end{aligned}
\ = \
\begin{aligned}
 \Qcircuit @C=1em @R=.7em @! R { \multiprepareC{1}{\Psi} & \ustick \rA \qw& \qw \\
 \pureghost {\Psi} & \ustick \rB \qw&\measureD{e}}
\end{aligned}\ .
\end{equation*}
If $\Psi$ is also pure, then we say
that it is a {\em purification} of $\rho$ and $\rB$ is called the {\em purifying system}. We denote by $D_{\rho}$ the set of all dilations of the state $\rho$, and by $P_\rho$ the set of all its purifications. 
\end{definition}
Trivially one has that $P_{\rho}\subseteq D_{\rho}$. Moreover, if $\Omega\in D_{\rho}$, then one has $D_{\Omega}\subseteq D_{\rho}$


The linear space $\st_{\R}(\rA)$ can be endowed with a metric structure by means of the following norm,
 which has an operational
meaning related to optimal discrimination schemes~\cite{PhysRevA.81.062348}.
\begin{definition}[Operational norm]
The norm of an element $\rho\in\st(\rA)_{\R}$ is defined as
\[
\norma{\rho}_{\rm op}\coloneqq\sup_{\{a_0,a_1\}\subseteq\eff(\rA)}(a_0-a_1|\rho),
\]
where $\{a_0,a_1\}$ is any binary observation test.
\end{definition}
This norm reduces to the usual trace-norm in the case of quantum theory. Moreover, it satisfies a monotonicity property, as stated in the following lemma
(see \cite{PhysRevA.81.062348} for the proof).
\begin{lemma}[Monotonicity of the operational norm]\label{lem:opnorm_monotonicity}
For any $\delta\in\st_{\R}(\rA)$ and $\tC\in\tr(\rA,\rB)$ the following inequality holds
\begin{equation}
\normop{\tC\delta}\leq\normop{\delta},
\end{equation}
with equality holding iff $\tC$ is reversible.
\end{lemma}

An OPT can enjoy different degrees of locality according to the following definition \cite{hardy2012limited}.
\begin{definition}[n-local discriminability] \label{def:nlocal_discr}
Let $n\leq m$. We say that that an OPT satisfies \emph{n-local discriminability} if the effects obtained as a conic combination of the parallel composition of effects $a_1,\dots,a_{n}$,
where $a_j$ is $k_j$-partite with $k_j\leq n$ for all values of j, are separating for $m$-partite states.
\end{definition}
We recall that a $k$-partite effect is an effect on a system that is obtained as the composition of $k$ systems.
The case of $n=1$, that is the case of classical and quantum theory, is simply called \emph{local discriminability}. In the latter case, this property tells us that state tomography can be accomplished by means of local measurements only, and it is equivalent to state that for any pair of
systems $\rA$, $\rB$ with size $D_{\rA}$ and $D_{\rB}$ respectively, the composite system $\rA\rB$ has size $D_{\rA\rB}=D_{\rA}D_{\rB}$. Notice that a
theory that is 
$n$-locally discriminable is also $n+1$-locally discriminable. Fermionic theory, real quantum theory and the bilocal classical theory considered here are examples of strictly bilocal theories, namely they satisfy definition \ref{def:nlocal_discr} with $n=2$ but not for $n=1$. An OPT with local discriminability also necessarily satisfies the following feature.

\begin{definition}\label{prop:atom_paralcomp_states}
An OPT satisfies atomicity of parallel composition of states if, for any pair of systems $\rA$ and $\rB$ and for all pair of atomic states 
$\phi\in\st(\rA)$ and $\psi\in\st(\rB)$, the composition 
$\phi\boxtimes\psi\in\st(\rA\rB)$ is atomic.
\end{definition}

This property is clearly satisfied by classical and quantum theory, and fermionic theory as well. The latter is a 
trivial example showing that the atomicity of parallel composition of states does not imply local discriminability. 
Moreover, examples of theories that violate atomicity of parallel composition of states can be constructed, and we 
will discuss one in detail in the present analysis.

We now introduce two properties that an OPT can have, and that are necessary for our discussion of the information 
content, whose definition will be recalled in the next subsection. These are strong causality and steering.
\begin{definition}[Causal theories]\label{ass:scausality}
An OPT satisfies {\em strong causality} if for every test $\{\tA_{i}\}_{i\in \mathsf{X}}$ and every collection of tests $\{\tB_{j}^{i}\}_{j\in\sY}$
labelled by $j\in\sY$, the collection of events $\{\tC_{i,j}\}_{(i,j)\in\sX\times\sY}$ with
\begin{equation}\label{eq:CondEv}
  \begin{aligned}
    \Qcircuit @C=1em @R=.7em @! R {&\ustick{\rA} \qw&\gate{\tC_{i,j}}& \ustick{\rC}
      \qw&\qw}
  \end{aligned}
  \quad\coloneqq\quad
  \begin{aligned}
    \Qcircuit @C=1em @R=.7em @! R {&\ustick{\rA}\qw&\gate{\tA_i}&\ustick{\rB} \qw &\gate{\tB_{j}^{i}}&\ustick{\rC} \qw
      &\qw}
  \end{aligned}\ 
\end{equation}
is a test of the theory.
\end{definition}
In other words, any theory that is strongly causal contains all the conceivable conditioned test, namely those in 
which the second test that is performed is chosen according to the outcome of the first one. This is a strong notion 
of causality, since  one can show that the above statement implies uniqueness of the deterministic 
effect~\cite{DAriano:2017aa,PhysRevA.81.062348}, that corresponds to a weaker form of causality defined as follows. 
\begin{definition}[Weakly causal theories]\label{ass:wcausality}
An OPT satisfies {\em weak causality} if for every system type $\rA$ there is a unique deterministic effect $e_\rA$.
\end{definition}

In a non deterministic theory (i.e.~a theory where at least one event of the trivial system is neither 0 nor 1) strong 
causality implies convexity of the set of events of every type---including, of course, sets of states of every system \cite{Chiribella2016,DAriano:2017aa}. An example of an OPT  that is weakly but not strongly causal is given in \cite{erba2023incompatibility}. For every system of this theory the set of states is a simplex whose vertices, corresponding to the pure states, are perfectly discriminable, just as in standard classical theory. The fact that the set of states is a simplex for every system implies that the theory is weakly causal (by theorem 1 in \cite{PhysRevA.101.042118}). However, the set of allowed tests between non-trivial systems only includes the composition of preparation and observation tests with permutations of subsystems of composite systems, and does not contain conditional tests, i.e., tests whose events are of the form \eqref{eq:CondEv}. This implies that such a theory is not strongly causal.

We now come to the second property of interest, which is the property of steering. 
Steering as we define it is actually a property of both classical and quantum theory. 
\begin{property}[Steering]\label{ass:steering}
Let $\rho\in\st(\rA)$ and $\{\sigma_{i}\}_{i\in \sX}\subseteq\st(\rA)$ be a refinement of $\rho$. Then there exist a system $\rB$, a state $\Psi\in\st(\rA\rB)$, and an
observation test $\{b_{i}\}_{i\in \sX}$ such that
\begin{equation*}
\  \
\begin{aligned}
 \Qcircuit @C=1em @R=.7em @! R { \prepareC{\sigma_{i}} & \ustick \rA \qw & \qw}
\end{aligned}
\ =\
\begin{aligned}
 \Qcircuit @C=1em @R=.7em @! R { \multiprepareC{1}{\Psi} & \ustick \rA \qw& \qw \\
 \pureghost {\Psi} & \ustick \rB \qw&\measureD{b_{i}}}
\end{aligned}\ ,
\qquad \forall i\in \sX.
\end{equation*}
\end{property}
Notice that the state $\Psi$ in the steering assumption must be a dilation of $\rho$ (see definition \ref{def:dilation}), as one can easily verify upon
summing over $i\in\sX$, and noticing that the coarse-graining of all the effects in an observation test yields a deterministic effect. Moreover, we are not requiring that the same dilation $\Psi$ steers all the possible 
refinements of $\rho$, differently from what happens in quantum theory. We remark indeed that in quantum theory a 
stronger version of steering holds. It states that any ensemble of a given state can be steered by means of a 
purification of that state. In more precise terms, given a state $\rho\in\st(\rA)$ and a purification 
$\Phi\in\purst(\rA\rB)$ of $\rho$, for any decomposition $\sum_{i\in\sX} p_i\sigma_i$ of $\rho$ there exists an 
observation test $\{ b_i\}_{i\in\sX}\subseteq\eff(\rB)$ such that
\begin{equation*}
\ p_i \
\begin{aligned}
 \Qcircuit @C=1em @R=.7em @! R { \prepareC{\sigma_{i}} & \ustick \rA \qw & \qw}
\end{aligned}
\ =\
\begin{aligned}
 \Qcircuit @C=1em @R=.7em @! R { \multiprepareC{1}{\Phi} & \ustick \rA \qw& \qw \\
 \pureghost {\Phi} & \ustick \rB \qw&\measureD{b_{i}}}
\end{aligned}\ .
\qquad \forall i\in \sX.
\end{equation*}
However, we do not need such a strong property for our purposes. In the following when referring to steering we 
will be referring to the property~\ref{ass:steering}.

Finally, with the following set of definitions we now make explicit what we mean by classical theory. First of all, we 
say that an OPT is {\em convex} when, for all systems of the theory, the set of states (and more generally of 
transformations) is convex. As we remarked above, this is always the case for non deterministic theories with strong 
causality.
\begin{definition}[simplicial theory]\label{def:simplicialtheories}
A {\em simplicial theory} is a finite-dimensional OPT where the extremal states of every system $\rA$ are the vertices 
of a $D_{\rA}$-simplex.
\end{definition}

\begin{definition}[Joint perfect discriminability]
Let $\rA$ be a system of the theory. A set of states $\{|\rho_i)_{\rA}\}_{i=1}^{n}$ is jointly perfectly discriminable if there exists an observation test $\{(a_i|_{\rA}\}$ such that:
\[
(a_i|\rho_j)_{\rA}\propto\delta_{ij}, \qquad \forall i,j\in\{1,\dots,n\}.
\]
\end{definition}

We now define classical theories only in the restricted context of weakly causal theories. The reason is that the 
definition would otherwise involve complications that are unnecessary for our present purposes. Indeed, the core of 
our analysis will deal with strongly causal theories, which are also weakly causal.

\begin{definition}[classical theories] \label{def:classicaltheories}
A {\em classical theory} is a theory where any set of pairwise non proportional atomic states is jointly perfectly discriminable.
\end{definition}


As it is clear, a classical theory in the sense of the above definition is not necessarily convex, despite the pure 
states are the extremal points of a simplex. However, when strong causality is assumed, the scenario is restricted to
convex theories, and thus classicality in that case implies that the set of states is a simplex whose extremal points 
are perfectly discriminable pure states.

While atomicity of parallel composition of states (definition \ref{prop:atom_paralcomp_states}) is always implied by local discriminability, the converse is generally not true, a trivial example being fermionic theory, that is strictly bilocal but state-purity preserving. Nevertheless, a converse holds for n-locally discriminable theories that are also simplicial, as a corollary of theorems 2 and 3 in \cite{PhysRevA.101.042118}: firstly, every simplicial theory satisfying n-local discriminability does not admit the presence of entangled states iff it satisfies atomicity of parallel composition of states (theorem 3 \cite{PhysRevA.101.042118}). But it is also true that a simplicial theory is locally discriminable iff there exist no entangled states (theorem 2 \cite{PhysRevA.101.042118}). Whereby, a simplicial n-locally discriminable theory enjoys atomicity of parallel composition of states iff it is actually locally discriminable.

In the following we will deal with strongly causal classical theories, where
the set of states of any system is a simplex whose extreme points, except for the null state (i.e. the state that yields zero on all the effects), are pure. Thus, atomicity of parallel composition for states is equivalent to purity of parallel composition of pure states.
 
\subsection{Digitizable theories and the information content}\label{subsec:infoContent}
In a general physical theory of information, information plays of course a central role. It is then reasonable to 
expect that one has a way to  quantify it, and this in turn has to come along with a unit. In classical 
and quantum theories, the output of any physical source is digitized in terms of bits and qubits, respectively. In 
particular, focusing on the case of classical information theory, any string of length $N$, say $\bi\coloneqq i_1\dots i_N$, 
where each $i$ is extracted from an alphabet of $d$ symbols, can be perfectly encoded on a distinct array of a 
suitable number of bits. The bit is indeed the reference system that we usually adopt as a unit for assessing the 
amount of information of a given classical source; an analogous role is played by the qubit in quantum information 
theory. Moreover, notice that one can actually choose any other system for the digitization of the outputs of a 
classical or quantum information source, with the only difference being that the entropy function, which numerically 
quantifies the information content, must be rescaled by a suitable multiplicative factor.

These simple observations can be thoroughly made into a requirement that a theory must abide by, thus setting the 
stage for a meaningful definition of compression rate in the OPT framework. This is formalized as follows.
\begin{definition}[Digitizability]\label{def:digitizability}
We say that an OPT is digitizable if there exists a system $\rB$ (called \emph{o-bit}) such that for any system $\rX$ there
exists $k<\infty$ and a pair of maps $\tC\in\tr_1(\rX\to\rB^{\boxtimes k})$ and 
$\tF\in\tr_1(\rB^{\boxtimes k}\to\rX)$ such that $\tF\circ\tC=\tI_{\rX}$. Moreover, if $\rB_1$ and $\rB_2$ are two such systems, then they are asymptotically equivalent.
\end{definition}
As we have already mentioned, finite-dimensional classical and quantum theory are digitizable. In particular, all the systems are asymptotically equivalent in the sense of definition \ref{def:asympeq}, and any system can serve as an o-bit. The OPT considered in \cite{erba2023incompatibility}, on the other hand, represents
an example of a theory that is not digitisable. This theory was constructed for a different purpose rather than studying digitizability, but it is not difficult to realize that it does not satisfy the latter property. As we have already mentioned, in this theory all systems are classical, in the sense that for any system the set of states is a simplex, and the vertices corresponding to non-null extremal states are perfectly jointly discriminable. Moreover,  for any integer $D$ there exists a type of system of size $D$. The set of allowed channels between non-trivial systems is restricted to contain compositions of deterministic states and effects with permutations of subsystems of composite systems. It is then clear that, given a system of type $\rA$, for every other system $\rB$ of different type it is not possible to find an integer $M$ and a left-reversible map of type $\rA\to\rB^{\boxtimes M}$. Indeed, such a left-reversible map should have the $\rA$ system as a subsystem at the output in order to be left-reversible, and given the uniqueness of the decomposition of integers in prime factors, this is impossible. Thus, the theory cannot be digitisable.

An \emph{information source} is characterized by a system $\rA$ and a state $\rho\in\st(\rA)$. Repeated uses of the 
source then generate a {\em message} that, generalizing the i.i.d. setting, has the form of a factorized state, 
$\rho^{\boxtimes N}$. Given a pair of maps $\tE\in\tr_1(\rA^{\boxtimes N}\rightarrow\rB^{\boxtimes M})$ and 
$\tD\in\tr_1(\rB^{\boxtimes M}\rightarrow\rA^{\boxtimes N})$, that we name \textit{encoding} and \textit{decoding} 
respectively, we need a figure of merit that establishes the goodness of the {\em codec} map $\tC\coloneqq\tD\tE$. Our choice comes as a consequence of the two following elementary observations: 
(i) the output message to which we have local access may be correlated with an ancilla, and (ii) it could be prepared 
in different ways. Since the aim is to preserve all the possible information gathered in the message, we use the 
following quantity
\[
\begin{aligned}
D(\rho^{\boxtimes N},\tC)& \\
\coloneqq\sup_{\Psi\in D_{\rho^{\boxtimes N}}} & \sup_{\{\Gamma_i\}\in \cR(\Psi)}\sum_{i}\normop{\tC\boxtimes\tI(\Gamma_i)-\Gamma_i}.
\end{aligned}
\]
In other words, it is an optimization of the average error over all the possible decompositions of any dilation of the message.
 We then define an \emph{$(N,\epsilon)$-compression scheme} as an encoding-decoding pair $(\tE,\tD)$ such that $D(\rho^{\boxtimes N},\tD\tE)<\epsilon$, and we denote by $E_{N,M,\epsilon}(\rho)$ the set of all the $(N,\epsilon)$-compression schemes. We now have all the ingredients for defining the information content of a source.
\begin{definition}[Information Content]\label{eq:IC-def}
Let $\rho\in\st_1(\rA)$ and 
\begin{equation}\label{eq:smallestM}
M_{N,\varepsilon}(\rho)\coloneqq\min\{ M:E_{N,M,\varepsilon}(\rho)\neq\emptyset\}. 
\end{equation}
We define the {\em smallest achievable compression ratio} for length $N$ to
tolerance $\varepsilon$ as follows
\begin{equation}\label{eq:inforate}
R_{N,\varepsilon}(\rho)\coloneqq\frac{\min\{ M:E_{N,M,\varepsilon}(\rho)\neq\emptyset\}}{N}.
\end{equation}
The {\em information content} of the state $\rho$ is defined as
\begin{equation}\label{eq:infocontent}
I(\rho)\coloneqq\lim_{\varepsilon\rightarrow 0}\limsup_{N\rightarrow \infty}R_{N,\varepsilon}(\rho).
\end{equation}
\end{definition}
The well-posedness of this definition has been proved in \cite{PhysRevA.105.052222}, as a direct consequence of the digitizability assumption. As it is clear, the definition of $I(\rho)$ strongly depends on the
choice of the figure of merit $D(\rho^{\boxtimes N},\tC)$. The one that we have introduced takes into account all possible refinements of any dilation of $\rho^{\boxtimes N}$. Alternatively, one can consider the pure convex refinements of any dilation of $\rho^{\boxtimes N}$ and define 
\[
\begin{aligned}
D^{\text{pur}}(\rho^{\boxtimes N},\tC)& \\
\coloneqq\sup_{\Psi\in D_{\rho^{\boxtimes N}}} & \sup_{\{p_i,\Psi_i\}\in \cP(\Psi)}\sum_{i}p_i\normop{\tC\boxtimes\tI(\Psi_i)-\Psi_i}.
\end{aligned}
\]
When the theory satisfies strong causality (property \ref{ass:scausality}), and if any state is proportional to a determistic one, $D$ and $D^{\rm pur}$ can be equivalently used to define the information content via the formula \eqref{eq:infocontent}. If a theory also satisfies steering (property \ref{ass:steering}), along with strong causality, one can use a figure of merit that is further simplified \cite{PhysRevA.105.052222}

\begin{equation}
\begin{aligned}
D^{\text{dil}}(\rho^{\boxtimes N},\tC) \coloneqq\sup_{\Psi\in D_{\rho^{\boxtimes N}}} \normop{\tC\boxtimes\tI(\Psi)-\Psi}.
\end{aligned}
\end{equation}

\section{A case study: the information content in Bilocal Classical Theory}

Bilocal Classical Theory (hereafter BCT) is a fully-fledged operational probabilistic theory that has been constructed with the aim of showing the independence of two fundamental notions: entanglement and complementarity. As we have already recalled, simplicial theories, in the sense of definition \ref{def:simplicialtheories}, admit of the presence of entangled states iff they do not satisfy the local discriminability principle (see Theorem 2 in \cite{PhysRevA.101.042118}), and this can be explicitly seen in the case of BCT. Here we outline the features of the theory that are relevant for the present work, in particular we recall the structure of state spaces, which are essentially the same of those of classical theory, and more importantly the way in which systems are composed in parallel (Postulate 1 and 2 of \cite{PhysRevA.101.042118}). And the characterization of channels. For a detailed account we refer the reader to \cite{PhysRevA.102.052216}.



From the point of view of the state space, BCT does
  not differ from standard classical theory. In BCT, for every integer
  $D>1$ it is assumed that there exists a unique system with size
  $D$. The set of states of each system is given by a simplex whose
  vertices are jointly perfectly discriminable, whence BCT is a convex
  classical theory, according to definition
  \ref{def:classicaltheories}. In particular, for any system $\rA$,
  the extremal states are the vertices of the simplex $\st(\rA)$, and
  the non-null extremal states, say
  $\{\rket{i}_{\rA}\}_{i=1}^{D_\rA}$, are perfectly discriminated by
  an observation test $\{\rbra{j}_{\rA}\}_{j=1}^{D_\rA}$,
  i.e.~$(j|i)_{\rA}=\delta_{ij}$. Moreover, given that the pure states
  correspond to the vertices of a simplex, any other state can be
  uniquely decomposed in terms of pure ones. In
  \cite{PhysRevA.101.042118} (see in particular theorem 1 therein) it
  has been shown that simplicial theories are necessarily weakly
  causal, whereby BCT satisfies weak causality.  

The salient features of BCT are a consequence of the
  behaviour of parallel composition of states. By theorem 4 in
  \cite{PhysRevA.101.042118}, in a simplicial\ theory the state space
  $\st(\rA\rB)$ of a composite system $\rA\rB$ is fully characterised
  by: i) a choice for the size $D_{\rA\rB}$ of the composite system
  $\rA\rB$; ii) an unambiguous labeling for the pure states of
  $\rA\rB$ as $\rket{(ij)_k}_{\rA\rB}$---for
  $i\in\{1,\dots,D_{\rA}\}$, $j\in\{1,\dots,D_{\rB}\}$, and $k$ in
  finite sets $I_{ij}$---and a choice of probability distributions
  $p_k^{ij}$ such that
\[
  \rket{i}_{\rA}\boxtimes\rket{j}_{\rB} = \sum _{k\in
    I_{ij}}p_k^{ij}\rket{(ij)_k}_{\rA\rB}.
\]
BCT is an explixit example of how the situation described above can
be consistently realised.  In the first place, the following rule for
the size of bipartite systems is postulated: for any two systems
$\rA,\rB$, the size of the composite system $\sys{AB}$ is given by:
\begin{align}\label{eq:dimensions}
D_{\rA\rB} =D_{\rB\rA}= \left\{
\begin{aligned}
&2D_{\rA}D_{\rB},&\mbox{if}\ \rA\neq\rI\neq\rB,\\
&D_{\rA},&\mbox{if}\ \rB=\rI.
\end{aligned}
\right.
\end{align}
Recalling that a theory satisfies local discriminability if and only
if $D_{\rA\rB}=D_{\rA}D_{\rB}$, the fact that BCT does not enjoy such
a property follows from the above compositional law
\eqref{eq:dimensions}. Actually BCT is strictly
bilocal~\cite{PhysRevA.102.052216}, since it satisfies the following
constraint on the dimension of tripartite systems (see theorem 2 in
\cite{PhysRevA.102.052216})
\[
  D_{\rA\rB\rC}=D_\rA D_\rB D_\rC + \Delta_{\rA\rB}D_{\rC} +
  \Delta_{\rB\rC}D_{\rA} + \Delta_{\rA\rC}D_{\rB}
\]
where $\Delta_{\rA\rB}\coloneqq D_{\rA\rB}-D_{\rA}D_{\rB}$, as one can easily
check using equation \eqref{eq:dimensions}.
The unambiguous labeling for the pure states of the composite system is made as follows: 
for any $i\in\{1,\dots,D_{\rA}\}$, $j\in\{1,\dots,D_{\rB}\}$, the set
$I_{ij}$ is given by $\{+,-\}$, and $p^{ij}_k\coloneqq\frac{1}{2}$. 
Let $\rI\neq\rA,\rB,\rC$. The set of pure states of any composite system $\rA\rB$ is then 
$\purst(\rA\rB)=\{\rket{(ij)_-}_{\rA\rB},\rket{(ij)_+}_{\rA\rB} \left.|\right. 1\leq i\leq D_{\rA},1\leq j\leq D_{\rB}\}$, 
so that for all pure states $\rket{i}_\rA\in\purst(\rA),\rket{j}_{\rB}\in\purst(\rB)$ the following parallel 
composition rule holds:
\begin{align}\label{eq:state_composition}
\begin{aligned}
\Qcircuit @C=1em @R=1em
{   &&\prepareC{i}&\ustick\rA\qw&\\
&\prepareC{j}&\ustick \rB\qw&\qw }
\end{aligned}
=
\frac{1}{2}\sum_{s=-,+}
\begin{aligned}
\Qcircuit @C=1em @R=1em
{   \multiprepareC{1}{(ij)_{s}}&\ustick \rA\qw&\\
\pureghost{(ij)_{s}}&\ustick \rB\qw& }
\end{aligned}.
\end{align}
When a third system is added the association satisfies the following law:
\begin{align}\label{eq:quadripartite}
((ij)_{s_1}k)_{s_2} = (i(jk)_{s_1s_2})_{s_1},
\end{align}
for all local indices $i,j,k$ and signs $s_1,s_2$.
As already mentioned, as a consequence of bilocality BCT admits of entangled states. More precisely, all the pure states of any bipartite system are entangled \cite{PhysRevA.102.052216}.

According to equation \eqref{eq:state_composition}, each time we compose in parallel a state with itself, there is an additional degree of freedom $s$ associated with a sign $+$ or $-$ which can be either of the two values with the same probability $\frac{1}{2}$. This entails that a message of length $N$ that is emitted from a source, described by $\sum_i p_i |i)_{\rA}=|\rho)_{\rA}\in\st_1(\rA)$, can be written as follows
\begin{equation}\label{eq:message}
|\rho^{\boxtimes N})_{\rA^N}=\sum_{\bi,\s}p_{\bi}\frac{1}{2^{N-1}}|\bi_{\s})_{\rA^N},
\end{equation}
where $\bi$ and $\s$ collectively denote the string of $N$ local indices and of $N-1$ signs respectively. Notice that,
according to the rule of Eq.~\eqref{eq:quadripartite}, the string of signs $\s$ depends on the order in which the $N$ systems 
are associated. If the order of composition changes, however, one just has a change in the string of signs 
$\s'=f(\s)$, according to Eq.~\eqref{eq:quadripartite}, which is immaterial since $\s$ is a dummy index and $f$ is an 
invertible function. Anyway, for the sake of clarity, we will ubiquitously adopt the convention that the expression 
in Eq.~\eqref{eq:message} refers to the composite system $(\ldots((\rA_1\rA_2)\rA_3)\ldots\rA_{N-1})\rA_N$. 

For our purpose it is of fundamental importance to know how transformations, and in particular channels, 
are characterized, since this establishes how much freedom we have in 
devising suitable compression schemes $(\tE,\tD)$. In
\cite{PhysRevA.102.052216} the authors first introduce the reversible
transformations as follows (see in particular postulate 3): for any
pair of non-trivial systems $\rA$ and $\rA'$ of the same size
$D_{\rA}$, $\tR$ is a reversible transformation if and only if there
exists a permutation of $D_{\rA}$ elements and a sign $\tau_i$ for every $i$ such
that for any non trivial $\rE$ and $\rket{(ij)_s}\in\purst(\rA\rB)$
one has
\begin{align}\label{eq:revtransfBCT}
\begin{aligned}
\Qcircuit @C=1em @R=.9em @! R { \multiprepareC{1}{(ij)_s} & \ustick  \rA \qw  & \gate {\tR} & \ustick  {\rA'} \qw &\qw  \\
\pureghost {(ij)_s} &  \qw  & \ustick \rE \qw &   \qw  &\qw }
\end{aligned}
\ = \
\begin{aligned}
\Qcircuit @C=1em @R=.7em @! R { \multiprepareC{1}{(\pi(i)j)_{\tau_i s}} & \ustick {\rA'} \qw & \qw \\
\pureghost {(\pi(i)j)_{\tau_i s}} & \ustick \rE \qw & \qw}
\end{aligned}.
\end{align}
Then, transformations are introduced as those admitting a reversible
dilation (postulate 4 in \cite{PhysRevA.102.052216}). More precisely,
for any pair of systems $\rA$ and $\rB$, $\tT\in\tr(\rA\to\rB)$ if and
only if there exists a reversible transformation $\tR\in\tr(\rB'\rA\to\rA'\rB)$, a state $\rket{\Sigma}_{\rB'}$ and an effect $\rbra{H}_{\rA'}$ such that 
\begin{align*}
\begin{aligned}
\Qcircuit @C=1em @R=.9em @! R { &\ustick{\rA} \qw & \gate{\tT} & \ustick{\rB} \qw }
\end{aligned}
\ = \
\begin{aligned}
\Qcircuit @C=1em @R=.7em @! R { \prepareC{\Sigma} & \ustick {\rB'} \qw & \multigate{1}{\tR}&  \ustick {\rA'} \qw& \measureD{H}\\
	&  \ustick \rA \qw & \ghost{\tR} &  \ustick {\rB'} \qw&\qw}
\end{aligned}.
\end{align*}
The above realisation of a BCT transformation is given in an implicit
form via its reversible simulation, whose action is explicitly defined
in equation \eqref{eq:revtransfBCT}. The following result, proved in
\cite{PhysRevA.102.052216} (see in particular lemma 3 and 4 in
appendix A), provides an explicit characterisation of transformations
and channels:
\begin{proposition}[Characterisation of BCT transformations]\label{prop:BCTchannels}
$\tC\in\tr(\rA\rightarrow\rB)$ iff for every $i\in\{1,\dots,D_{\rA}\}$ there exists a set $\{\lambda_{m,\tau}^{(i)}\}_{(m,\tau)\in I^{(i)}}$,
with $I^{(i)}\subseteq\{1,\dots,D_{\rB}\}\times\{+,-\}$, $\lambda_{m,\tau}^{(i)}>0$ for every $(m,\tau)\in I^{(i)}$ and $\sum_{(m,\tau)\in I^{(i)}}\lambda^{(i)}_{m,\tau}\leq 1$, such that the following holds for all $\rket{(ij)_s}_{\rA\rE}\in\purst(\rA\rE)$:
\begin{align*}
\begin{aligned}
\Qcircuit @C=1em @R=.7em @! R { \multiprepareC{1}{(ij)_s} & \qw  &\ustick  \rA \qw &\qw &\qw \\
\pureghost {(ij)_s} & \ustick \rE  \qw  & \gate {\tC} & \ustick \rF  \qw  &\qw }
\end{aligned}
\ =\sum_{m,\tau}\lambda_{m\tau}^{(j)} \
\begin{aligned}
\Qcircuit @C=1em @R=.7em @! R { \multiprepareC{1}{(im)_{\tau s}} & \ustick \rA \qw & \qw \\
\pureghost {(im)_{\tau s}} & \ustick \rF \qw & \qw}
\end{aligned}.
\end{align*}
If the transformation is determinsitic, then $\{\lambda_{m,\tau}^{(i)}\}_{(m,\tau)\in I^{(i)}}$ is a probability distribution for every $i$.
\end{proposition}

We can now analyse the behaviour of the information
  content in BCT by applying the formalism developed in subsection
  \ref{subsec:infoContent}. The first fact that must be checked is
that BCT is a digitizable theory, in the sense of
definition~\ref{def:digitizability}. This is a mandatory step that
allows for a meaningful operational definition of information content,
as discussed at the end of the previous section. The proof is given in
appendix \ref{app:proofs} (lemma \ref{lem:BCTdigitizability}). In BCT,
as in classical and quantum theory, any type of system can serve as
o-bit. Here we choose the type of system with $D=2$, that we will call
\emph{bibit} from now on.

The major obstacle in computing the information content of a given state is the complexity of the figure of merit. The 
greater is the set of states on which we must validate the compression scheme, the more difficult is to devise one 
that works as we wish. However, for any state $\rho\in\st_1(\rA)$ of BCT, one can find a dilation $\Pi\in\st(\rA\rE)$ 
from which we can compute all the other ones by applying a suitable channel on the ancillary system $\rE$, as it is 
stated in the following proposition.
\begin{proposition}\label{prop:motherdilation}
Let $\sum_ip_i|i)_{\rA}=|\rho)_{\rA}\in\st_1(\rA)$ and let $\sum_{ijs}q_{ijs}|(ij)_s)_{\rA\rF}=|\Psi)_{\rA\rF}\in\st_1(\rA\rF)$ be a dilation of $\rho$. Let $\rE\cong\rA$ and $\sum_{iks}p_{iks}|(ik)_s)=|\Pi)_{\rA\rE}\in\st_1(\rA\rE)$ be the dilation of $\rho$ with joint probability distribution defined as
\[
p_{iks}\coloneqq\left\{
\begin{aligned}
&\delta_{ik}p_i,&s=+; \\
&0, &s=-.
\end{aligned}
\right.
\]
Then, there exists a channel $\tC\in\tr_1(\rE,\rF)$ such that 

\begin{align*}
\begin{aligned}
\Qcircuit @C=1em @R=.7em @! R { \multiprepareC{1}{\Psi} & \ustick \rA \qw & \qw \\
\pureghost {\Psi} & \ustick \rF \qw & \qw}
\end{aligned}
\ = \
\begin{aligned}
\Qcircuit @C=1em @R=.7em @! R { \multiprepareC{1}{\Pi} & \qw  &\ustick  \rA \qw &\qw &\qw \\
\pureghost {\Pi} & \ustick \rE  \qw  & \gate {\tC} & \ustick \rF  \qw  &\qw }
\end{aligned}.
\end{align*}
\end{proposition} 
The proof of this proposition and of all the other results of this and the subsequent section are given in appendix \ref{app:proofs}.
The above statement, along with the fact that BCT satisfies the strong causality principle and the steering property, straightforwardly implies the following 
proposition, which drastically simplifies the task of devising a compression scheme.
\begin{proposition}\label{prop:fig_merit}
Given a state $\rho\in\st_1(\rA)$ and a compression scheme $(\tE,\tD)$ for a message of length $N$, the figure of merit can be computed according to the following formula
\begin{multline}\label{eq:BCTfom}
\tilde{D}(\rho^{\boxtimes N},\tC)= \\ 
\sum_{\is}\frac{1}{2^{N-1}}p_{\bf{i}}\normop{(\tC-\tI_\rA^N)\boxtimes\tI_{\rA^N})|(\is\is)_+)_{\rA^N\rA^N}}.
\end{multline}
In other words, defining $\tilde{I}(\rho)$ analogously to $I(\rho)$ (equation \eqref{eq:infocontent}) by replacing $D(\rho^{\boxtimes N},\tC)$ with 
 $\tilde{D}(\rho^{\boxtimes N},\tC)$, it holds that $I(\rho)=\tilde{I}(\rho)$.
\end{proposition}

\begin{theorem}\label{thm:main_thm}
Let $\rA$ be a BCT system and $\sum_{i=1}^{D_{\rA}}p_i|i)_{\rA}=\rho\in\st_1(\rA)$. Then 

\begin{equation}
I(\rho)=\frac{H(\p)+1}{2}.\label{eq:BCTinfocontent}
\end{equation}
\end{theorem}
A first corollary of formula \eqref{eq:BCTinfocontent} is that the information content is additive in BCT. Indeed, if $\sum_ip_i\rket{i}_{\rA}=\rket{\rho}_{\rA}\in\st(\rA)$ and 
$\sum_jq_j\rket{j}_{\rB}=\rket{\sigma}_{\rB}\in\st(\rB)$ then the factorized state is given by 
\[
\rket{\rho}_{\rA}\boxtimes\rket{\sigma}_{\rB}=\sum_{i,j,s}\frac{1}{2}p_i q_j \rket{(ij)_s}_{\rA\rB},
\]
a straightforward application of equation \eqref{eq:BCTinfocontent} then gives 

\[
\begin{aligned}
I(\rho\boxtimes\sigma)& =\frac{[H(\p)+H(\q)+1]+1}{2} \\
	&= \frac{H(\p)+1}{2}+\frac{H(\q)+1}{2} \\
	&= I(\rho)+I(\sigma).
\end{aligned}
\]
We then deduce that atomicity of parallel composition of states is not a necessary condition for the additivity of information content when factorized states are considered. In particular, not even local discriminability is a necessary condition for additivity of $I(\rho)$. The latter fact, however, was already known from 
\cite{Perinotti2023}, where it is proven that the information content of a fermionic source is given by the von Neumann entropy of the state representing the source, just as in the case of quantum theory, and fermionic theory is a bi-local theory (see \cite{doi:10.1142/S0217751X14300257,2014fermionic}).

Another interesting feature of the information content in this theory is that it is strictly positive for all states 
of any system. In particular, the Shannon entropy of any sharp probability distribution is vanishing, whence, for any 
system $\rA$, it holds that
\begin{equation}
I(\phi)=\frac{1}{2}, \qquad \forall \phi\in\purst(\rA).
\end{equation} 
This is in contrast with what we know from classical, quantum and fermionic theory, where the information content is vanishing if and only if the state
is pure. In this respect, one is led to give the notion of purity an operational meaning by saying that we have maximal knowledge about a physical system
whenever it is in a pure state. In \cite{PhysRevA.105.052222} (proposition V.2) It has been shown that atomicity of parallel composition of states (definition \ref{prop:atom_paralcomp_states}) and the uniqueness of purifications up to reversible transofrmations on the ancillary system, are sufficient conditions for this interpretation. Notably, a converse is also true \cite{PhysRevA.105.052222} (proposition V.2), namely that if $I(\phi)=0$ whenever the state $\phi$ is pure then state purity is preserved under parallel
composition. Here we explicitly see that in a theory that violates atomicity of parallel composition of states also pure states can have non vanishing information content.

In reference \cite{Short:2010kt}, the authors address the problem of noiseless coding within the generalised probabilistic theories framework. Under certain assumptions, they prove (see in particular theorem VIII.I in \cite{Short:2010kt}) that the measurement entropy $S_1(\rho)$ is a lower bound for the rate of the compression task they define in section VIII. At first glance, this result seems to be in contradiction with theorem \ref{thm:main_thm}. However,
we remark that two assumptions are made in order to prove such a result, that are not satisfied by the BCT. The first one is that the composition of fine-grained measurements yields a fine-grained measurement, where the latter is defined in \cite{Short:2010kt} as a measurement such that all its effects are atomic. BCT does not satisfy this property, because of the violation of the atomicity of parallel composition of states, which also entails that the composition of two fine-grained effects is no longer fine-grained. The second assumption is that, if the dimension of the state space of system $\rA$ is $d$, then the dimension of the state space of the composite system $\rA^{\boxtimes N}$ is $d^N$ (here, with the word \emph{dimension}, we are referring to the notion introduced in \cite{Short:2010kt}, that is the maximum number of elements of a fine-grained observaiton test). In the case of BCT, for a system $\rA$ with size $D_{\rA}$ (recall that the \emph{size} is the dimension of the linear space $\st_{\R}(\rA)$), the dimension is $d_{\rA}=D_\rA$, so that for $\rA^{\boxtimes N}$ we find $d_{\rA^{\boxtimes N}}=2^{N-1}D_{\rA}^N>D_{\rA}^N$ for $N>1$, whence also this assumption is violated in our case. Therefore, our result is not in contradiction with theorem VIII.I presented in reference \cite{Short:2010kt}, since BCT is outside the domain of validity of that theorem. Nevertheless, this result has its relevance for theories in which parallel composition preserves atomicity, in particular in the case of PR-boxes, which represents a very interesting case to investigate.

\section{Comparing the information content with entropies}
Entropic-like quantities have been introduced in the context of Generalized Probabilistic Theories in terms of the Shannon entropy function \cite{Barnum:2010gy,Short:2010kt,KIMURA2010175}. 
We report just below their definitions for the sake of completeness.
\begin{definition}\label{def:GPTentropies}
Let $\rho\in\st_1(\rA)$ for some system $\rA$, denote  by $\cO^{\rm at}(\rA)$ the set of atomic observation tests of
$\rA$ (i.e.~observation tests whose elements are atomic effects).
We then define 
\begin{align}
S_1(\rho)&\coloneqq\inf_{\{a_j\}\in\cO^{\rm at}(\rA)}H(J) \\
S_2(\rho)&\coloneqq\sup_{(p_i,\{\phi_i\})\in\cP(\rho)}\sup_{\{a_j\}\in\cO^{\rm at}(\rA)}H(I:J)\\
S_3(\rho)&\coloneqq\inf_{(p_i,\{\phi_i\})\in\cP(\rho)}H(I)
\end{align}
where in the Shannon entropies and the mutual information on the r.h.s.~of the above equations refer to the random variables $I,J$ distributed 
according to the joint probability distribution $q(i,j)\coloneqq p_i(a_j|\phi_i)_{\rA}$.
For any $i=1,2,3$ we introduce the regularized version of $S_i$ as follows:
\[
S^{\rm reg}_i(\rho)\coloneqq\limsup_{N\rightarrow \infty}\frac{S_i(\rho^{\boxtimes^N})}{N}.
\]
\end{definition} 

Elementary properties such as concavity, strong concavity and subadditivity have been studied for these quantities; moreover, they have been proven to be not equivalent. As we have seen in the foregoing section, the relation between the information content of a state and the Shannon entropy of the associated probability distribution is not trivial, as a consequence of the violation of atomicity of parallel composition of states. The latter property has also a remarkable consequence on the
behaviour of the regularized entropies with respect to their single-system counterparts, which is stated in the proposition below.

\begin{proposition}\label{prop:regentr_noatomicity}
In any classical theory, for any system $\rA$ and $\sum_ip_i\rket{i}_{\rA}=\rket{\rho}_{\rA}\in\st(\rA)$, one has $S_i(\rho)=H(\p)$. Moreover, whenever atomicity of parallel composition of states (definition \ref{prop:atom_paralcomp_states}) is violated, there exists a state $\Sigma\in\st(\rC)$ for some system $\rC$
such that the strict inequality $S_i^{\rm reg}(\Sigma)>S_i(\Sigma)$ holds for any $i=1,2,3$.
\end{proposition}

As a corollary of the above proposition, we can also notice that, for classical theories violating 
property~\ref{prop:atom_paralcomp_states}, each entropy, in addition to being superadditive, also violates additivity 
when factorized states are considered. Indeed, if there exist $\rket{i}_{\rA}\in\purst(\rA)$ and 
$\rket{j}_{\rB}\in\purst(\rB)$ such that $\rket{i}\boxtimes\rket{j}$ is mixed, we immediately see that 
\[
S_i(\rket{i}\boxtimes\rket{j})>0=S_i(\rket{i})+S_i(\rket{j}).
\]
where $S_i$ is given by the Shannon entropy of the respective decompositions, according to proposition \ref{prop:regentr_noatomicity}.

While in classical and quantum theory all the $S_i$'s and their regularized version collapse to the Shannon and von 
Neumann entropies, respectively, thus boiling down to the same operational interpretation given by the noiseless 
coding theorems, much less is known about their operational meaning in a general theory. In BCT the regularized 
entropies are related to the Shannon entropy of the state according to the following proposition.
\begin{proposition}\label{prop:entropiesBCT}
Let $\sum_ip_i|i)_{\rA}=|\rho)_{\rA}\in\st_1(\rA)$, then $S_i^{\rm reg}(\rho)=S_i(\rho)+1=H(\p)+1$.
\end{proposition}
A remarkable corollary of this proposition is that, in general, none of the $S_i$'s nor the $S_i^{\rm reg}$'s can be understood as the minimal compression rate.
Notice that, while the additivity property for factorized states is violated by all the entropies, it is satisfied by the regularized versions. The result of proposition 
\ref{prop:entropiesBCT} is by far intuitive if we think of the particular compositional rule on states that BCT satisfies. Indeed, at the level of single systems, there is no difference
with respect to standard classical theory, and this is true for any classical theory that does not satisfy atomicity of parallel composition of states. The effect of this violation shows up when we consider $N$ copies of the same state, the latter operation giving an extra flat bit, one for each additional copy of the original state, and the appearence of this extra bit is captured by the regularized entropies. The factor 2 in the information content can also be intuitively expected, since also when we compose bibits we get additional space that can be used to allocate the message. Therefore, the departure of $I(\rho)$ from $S^{\rm reg}_i(\rho)$ can be essentialy ascribed to the weird compositional rule for systems.

We notice that the results of theorem \ref{thm:main_thm} and proposition \ref{prop:entropiesBCT} are consistent with the following general bound
from \cite{PhysRevA.105.052222} (in particular, see the proof of lemma V.2 and the result of lemma V.1),

\[
I(\rho)\geq \frac{S_2(\rho)}{\log_2 D},
\]
where $D$ is a costant such that $D_{\rB^{\boxtimes M}}\leq k D^M$ for some $k$. Indeed, in the case of BCT $D_{\rB^{\boxtimes M}}=\frac{1}{2}4^M$, whence
$D=4$. Actually, in the present case such a bound is saturated with $S_2(\rho)$ replaced by $S_2^{\rm reg}(\rho)$.

\section{Conclusion}
We have presented a full computation of the information content in a bilocal classical theory. The result is given in 
terms of the Shannon entropy of the probability distribution defining the state, namely, if 
$\rket{\rho}=\sum_i p_i \rket{i}$, it turns out that $I(\rho)=(H(\p)+1)/{2}$. The definition of $I(\rho)$ is 
hardly amenable to direct calculation in general theories. However, in the special case of BCT, the calculation is simplified by the fact that, for any state, there exists a 
``mother'' dilation from which we can obtain all the other ones by applying a suitable channel on the ancilla. With 
respect to standard classical theory, the information content shows two differences that can be both ascribed to the 
violation of atomicity of parallel composition of states in BCT: i) there is an overhead given by +1 in the numerator. 
This is due to the appearance of a bit each time that we compose in parallel a new copy  of the same state. Since each 
bit is uniformly distributed we end up with the maximum overhead, that is indeed +1; ii) there is a factor 2 in the 
denominator. This follows from the fact that, when we compose bibits into registers, their dimension is given by 
the formula $D_{\rB^{\boxtimes M}}=2^{M-1}D_{\rB}^M$, thus the room for allocating messages per single bibit is almost ``double'' with respect to 
the size of the register. Notice that the factor 2 in the denominator is
then related to meaning of information content in the specific theory at hand, where the elementary systems for physical encoding are bibits. If we had to evaluate the {\em classical} information content, i.e.~the ability of the source to encode classical information, then the regularized mutual information would be the right quantifier.

We have already noticed that the information content in BCT is additive on factorized states---i.e.~states of 
the form $\rho\boxtimes\sigma$---and this means that atomicity of parallel composition is not a necessary condition
for additivity. A question that remains open is under which hypotheses, given two states 
$\rho,\sigma\in\st_1(\rA)$, one is able to prove that $I(\rho\boxtimes\sigma)=I(\rho)+I(\sigma)$. It might be also the 
case that additivity is a feature of $I(\rho)$ that holds in full generality, as it would be desirable for a measure of the information content. The results of the present work seem to suggest 
that atomicity of parallel composition plays a marginal role for this property.

Along with the information content we have also analyzed the behaviour of three different entropic functions that have 
been considered in the literature. At the level of single system there is no difference with respect to standard 
classical theory, and they all coincide with the Shannon entropy of the state $\rho$. As a consequence, any classical 
theory is monoentropic~\cite{Barnum:2010gy}. The regularized entropies are clearly sensible to the extra bit that 
arises when systems are composed, and they all turn out to be equal to $H(\p)+1$. This result then establishes the 
existence of theories of information where none of the proposed generalizations of entropy can be interpreted as the 
information content of the source, and neither their regularized versions do. Remarkably, this is true in a thory that 
is monoentropic. The departure of the regularized version from the single-system counterpart is not peculiar of the 
BCT, but it actually takes place in any classical theory (in the sense of the definition \ref{def:classicaltheories}) 
whenever atomicity of parallel composition of states is violated, as proved in 
proposition~\ref{prop:regentr_noatomicity}.

As we have already noticed at the end of the last section, in the case of BCT there is a relation between the 
regularized entropy $S_2^{\rm reg}(\rho)$ and the information content of the following form
\begin{equation}\label{eq:conjecture}
I(\rho)=\frac{S_2^{\rm reg}(\rho)}{\log_2 D},
\end{equation}
where $D$ is a constant such that the relation $D_{\rB^{\boxtimes M}}\leq kD^M$ holds for some $k$. In the case of BCT $\rB$ is the bibit, and this relation is saturated with $k=\frac{1}{2}$ and $D=4$, whence the equation above trivially holds. One might be tempted to conjecture that a result of this form holds for any classical
theory, but it is not difficult to realize that this is not the case. Let us consider a classical theory 
with only one type of system, say the bit (whose size is 2), satisfying local
discriminability and, consequently, atomicity of parallel composition. Now, restrict the allowed tests
of the theory to be preparation tests, observation tests, and all possible permutations of the bits when more bits are composed in parallel. It is easy to realize that there are
no protocols that allow one to compress a source represented by a mixed state of a single bit, so that $I(\rho)=1$. But proposition \ref{prop:regentr_noatomicity}
implies that $S_i(\rho)=h(p)$ where $h(p)$ is the binary entropy of the bit state $\rho$, and so is for $S_i^{\rm reg}(\rho)$ by local discriminability,
whence the conjecture is false. 

The main lesson that we learn from the results presented so far is that a treatment of the notion of information 
content, from a foundational point of view, cannot ignore the compositional structure of a physical theory, as we have 
seen that the latter heavily marks its behaviour also in a classical theory. Moreover, as we have just argued, also 
the allowed transformations play a significant role, as they might severely 
restrict the freedom of compressing. 

We leave open for further studies what happens if one considers non-local classical theories with a no-restriction 
hypothesis on the allowed transformations. It might be the case that a relation very similar in form to 
equation~\eqref{eq:conjecture} holds. 
Aiming at doing a first step beyond the the non-local classical scenario, one could ask how the present formalism can be applied to other known non-local theories, in order to investigate possible deviations of the information content from the expected behaviour in the absence of local discriminability. Recently, the authors studied information compression in fermionic theory \cite{Perinotti2023}, which is a strictly bilocal theory, showing that the information content equals the von Neumann entropy of the fermionic state, and a similar analysis could be carried out for real quantum theory, which also violates the local discriminability. Another model that could highlight relevant features of the information content is that of Ref. \cite{Barnum2020composites}, where finite-dimensional real, complex and quaternionic quantum theory have been unified within a single category, where complex quantum systems compose in a non-trivial way.
Another interesting question is what happens in the case of the Popescu-Rohrlich 
boxes~\cite{Popescu1994}. It is possible that one of the entropies is equal to the information content. On the one hand, we 
already know that $I(\rho)\geq S_2(\rho)$, therefore the missing part is achievability, namely the direct part of a 
noiseless coding theorem, which would establish $I(\rho)=S_2(\rho)$. The question in the case of PR-boxes is particularly relevant in light of the fact that the three entropies are known to be inequivalent in such a context.

Answering fundamental questions about information content and entropic functions is a first step toward the formulation of area laws and holographic principle beyond the standard quantum scenario. Both area laws and the holographic principle rely on the notion of entropy, which in turn is operationally interpreted as a quantifier of uncertainty (or information content). Understanding to what extent those laws can be generalized, independently of the nature of the systems involved, may shed new light on the relation between microscopic and large-scale physical phenomena in terms of localization and flow of information \cite{doi:10.1063/1.531249,RevModPhys.74.825,eisert2010colloquium}.

\section*{Aknowledgments}
P.P. and L.V. acknowledge financial support from European Union - Next Generation EU through the PNRR MUR Project No. PE0000023. A.T. acknowledges the financial support of Elvia and Federico Faggin Foundation (Silicon Valley Community Foundation Project ID No. 2020- 214365).

\bibliographystyle{unsrturl}
\bibliography{biblio-ic}
 
\appendix

\section{Proofs of the main results}\label{app:proofs}

\begin{lemma}\label{lem:BCTdigitizability}
BCT is a digitizable theory.
\end{lemma}
\begin{proof}[\bf Proof]
We show that any system $\rB$ of any size $D_{\rB}$ can serve as o-bit for BCT. Indeed, let $\rA$ be any other system of the theory, denote by $D_{\rA}$ its size and set
\begin{equation}
k=\left\lceil\log_{2D_{\rB}}2D_{\rA}\right\rceil. \label{eq:leastinteger}
\end{equation}
Then, let $\hf:\{1,\cdots,D_{\rA}\}\rightarrow\{1,\cdots,D_{\rB}\}^{k}\times\{+,-\}^{k-1}$ be an injiective function.
To denote the pair $({\bf j,\bf s})\in\{1,\cdots,D_{\rB}\}^{k}\times\{+,-\}^{k-1}$ we use the shortened notation $\js$.
The action of the encoding 
$\tE$ and the decoding $\tD$ are defined by two set of probability distributions, $\lambda^{(i)}_{\js,\tau}$ and $\mu^{(\js)}_{i,\tau}$ respectively. We 
then set $\lambda_{\hf(i),+}^{(i)}=1$ $\forall i \in \{1,\cdots,D_{\rA}\}$ for the encoding. For the decoding, $\forall \js$ such that $\exists!i$ satisfying $\hf(i)=\js$ we define $\mu^{(\js)}_{i,+}=1$, while for every other $\js$ we can freely choose any probability distribution at our wish. It is now easy to realize that, for any $i,j$ and $\rE$, the following holds

\begin{align}\label{eq:digiteq}
\begin{aligned}
\Qcircuit @C=1.2em @R=1em @! R { \multiprepareC{1}{(ij)_s} & \ustick {\rA}  \qw  & \gate {\tD\tE} & \ustick {\rA}  \qw  &\qw \\
\pureghost {(ij)_s} & \qw  &\ustick  {\rE} \qw &\qw &\qw }
\end{aligned}
\ = \
\begin{aligned}
\Qcircuit @C=1.8em @R=1em @! R { \multiprepareC{1}{(ij)_s} & \ustick {\rA} \qw & \qw \\
\pureghost {(ij)_s} & \ustick {\rE} \qw & \qw}
\end{aligned}\ .
\end{align}

Now consider two systems of BCT, say $\rB_1$ and $\rB_2$. For any integer number $k_1$ of systems $\rB_1$, the minimal number $M^{\rm min}_2(k_1)$ of $\rB_2$ that are needed to perform the encoding in a perfectly recoverable way as in equation \eqref{eq:digiteq} (with $\rA=\rB_1^{\boxtimes k_1}$) is given by formula \eqref{eq:leastinteger}
\[
M_2^{\rm min}(k_1)=\left\lceil k_1\log_{2D_{\rB_2}}2D_{\rB_1}\right\rceil,
\] 
and similarly
\[
M_1^{\rm min}(k_2)=\left\lceil k_2 \log_{2D_{\rB_1}}2D_{\rB_2}\right\rceil.
\] 
Therefore
\begin{align*}
\lim_{k_1\rightarrow\infty}\frac{M_2^{\rm min}(k_1)}{k_1}&=\log_{2D_{\rB_2}}2D_{\rB_1}\\
&=\left(\lim_{k_2\rightarrow\infty}\frac{M_1^{\rm min}(k_2)}{k_2}\right)^{-1},
\end{align*}
as required.
\end{proof}

\begin{proof}[\bf Proof of proposition \ref{prop:motherdilation}]
Any deterministic transformation  $\tC\in\tr_1(\rE,\rF)$ acts on the set of pure states of $\rA\rE$ as follows

\begin{align*}
\begin{aligned}
\Qcircuit @C=1em @R=.7em @! R { \multiprepareC{1}{(ik)_s} & \qw  &\ustick  \rA \qw &\qw &\qw \\
\pureghost {(ik)_s} & \ustick \rE  \qw  & \gate {\tC} & \ustick \rF  \qw  &\qw }
\end{aligned}
\ =\sum_{j,\tau}\lambda_{j\tau}^{(k)} \
\begin{aligned}
\Qcircuit @C=1em @R=.7em @! R { \multiprepareC{1}{(ij)_{\tau s}} & \ustick \rA \qw & \qw \\
\pureghost {(ij)_{\tau s}} & \ustick \rF \qw & \qw}
\end{aligned} \qquad \forall i,k,s
\end{align*}
where, for any $k$, $\lambda_{j\tau}^{(k)}$ is a probability distribution. Determining $\tC$ means to determine $\lambda^{(k)}_{j\tau}$ for every $k$. By applying the channel $\tC$ to the state $|\Pi)_{\rA\rE}$ defined in the statement, expanding and re-collecting the signs suitably, one finds

\begin{align*}
\begin{aligned}
\Qcircuit @C=1em @R=.7em @! R { \multiprepareC{1}{\Pi} & \qw  &\ustick  \rA \qw &\qw &\qw \\
\pureghost {\Pi} & \ustick \rE  \qw  & \gate {\tC} & \ustick \rF  \qw  &\qw }
\end{aligned}
\ =\sum_{i,j,s}\left(\sum_{\tau,k}p_{ik\tau}\lambda_{j(s\tau)}^{(k)}\right) \
\begin{aligned}
\Qcircuit @C=1em @R=.7em @! R { \multiprepareC{1}{(ij)_{s}} & \ustick \rA \qw & \qw \\
\pureghost {(ij)_{s}} & \ustick \rF \qw & \qw}
\end{aligned}
\end{align*}
Since the set of states is a simplex, the pure ones are affinely independent, and this entails that the condition 
$\tI_{\rA}\boxtimes\tC|\Pi)_{\rA\rE}=|\Psi_{\rA\rF})$ is satisfied iff the following equation holds for any $i,j$ and $s$

\begin{equation}
q_{ijs}=\sum_{\tau,k}p_{ik\tau}\lambda_{j(s\tau)}^{(k)}.
\end{equation}
For $p_{ik\tau}$ as in the statement
\begin{equation}
q_{ijs}=\sum_{\tau,k}p_{ik\tau}\lambda_{j(s\tau)}^{(k)}=\sum_{k}\delta_{ik}p_i\lambda_{js}^{(k)}=p_i\lambda_{js}^{(i)}.
\end{equation}
Since $\sum_{j,s}q_{ijs}=p_i$, for those indices $i$ such that $p_i=0$ the equation trivially holds whatever the choice of $\lambda^{(i)}_{js}$ is, while for those $i$'s such that $p_i>0$, $\lambda_{js}^{(i)}=\frac{q_{ijs}}{p_i}$ is an admissible solution.
\end{proof}

\begin{proof}[\bf Proof of proposition \ref{prop:fig_merit}]
Given any state $\rho\in\st_1(\rA)$, define the following dilation of $\rho^{\boxtimes N}$ (see equation \eqref{eq:message}) . 
\[
\rket{\Pi^N}\coloneqq\sum_{\sti,\s}\frac{1}{2^{N-1}}p_{\sti}\rket{(\is\is)_+}_{\rA^N\rA^N}.
\]
Moreover, for the sake of clarity let us set
\[
\begin{aligned}
E_{N,M,\varepsilon}^{\rm pur}(\rho) &= \{(\tE,\tD)|D^{\rm pur}(\rho^{\boxtimes N},\tC)<\varepsilon\}, \\
E_{N,M,\varepsilon}^{\rm dil}(\rho) &= \{(\tE,\tD)|D^{\rm dil}(\rho^{\boxtimes N},\tC)<\varepsilon\}, \\
\tilde{E}_{N,M,\varepsilon}^{\rm }(\rho) &= \{(\tE,\tD)|\tilde{D}(\rho^{\boxtimes N},\tC)<\varepsilon\}.
\end{aligned}
\]
Denote by $I^{\rm pur}(\rho)$, $I^{\rm dil}(\rho)$, $\tilde{I}(\rho)$ the corresponding quantities defined in a way analogous to \eqref{eq:infocontent}. 
The purpose is to show that $I(\rho)=\tilde{I}(\rho)$. We already know that $I(\rho)=I^{\rm pur}(\rho)=I^{\rm dil}(\rho)$. Since it is clear that 
$\tilde{D}(\rho^{\boxtimes N},\tC)\leq D^{\rm pur}(\rho^{\boxtimes N},\tC)$, it follows that $E^{\rm pur}_{N,M,\varepsilon}(\rho)\subseteq \tilde{E}_{N,M\varepsilon}(\rho)$ and thus $\tilde{I}(\rho)\leq I^{\rm pur}(\rho)=I(\rho)$. On the other hand, 
by proposition \ref{prop:motherdilation}, lemma \ref{lem:opnorm_monotonicity} (monotonicity of the operational norm with respect to channels), and the triangle inequality one also has, for any $\Psi\in D_{\rho^{\boxtimes N}}$ and $(\tE,\tD)\in \tilde{E}_{N,M,\varepsilon}(\rho)$
\[
\begin{aligned}
&\normop{\tC\boxtimes \tI_{\rE}\rket{\Psi}_{\rA^N\rE}-\rket{\Psi}_{\rA^N\rE}} = \\
= & \normop{\tI_{\rA^N}\boxtimes\tA(\tC\boxtimes \tI_{\rE}\rket{\Pi^N}_{\rA^N\rE}-\rket{\Pi^N}_{\rA^N\rE})} \\
\leq & \normop{\tC\boxtimes \tI_{\rE}\rket{\Pi^N}_{\rA^N\rE}-\rket{\Pi^N}_{\rA^N\rE}} \leq  \\
\leq & \tilde{D}(\rho^{\boxtimes N},\tC)
\end{aligned}
\]
whence $D^{\rm dil}(\rho^{\boxtimes N},\tC)\leq\tilde{D}(\rho^{\boxtimes N},\tC)$. This implies the inequality $I^{\rm dil}(\rho)\leq\tilde{I}(\rho)$. We have then proved that $I(\rho)=\tilde{I}(\rho)$.
\end{proof}

In order to relate the information content $I(\rho)$ to the Shannon entropy $H(\p)$ of the probability distribution defining the unique pure state decomposition of $\rho$, we use standard techniques from classical information theory. In particolar we remind the reader of the notion of typical
strings and of typical set, along with some properties of the latter.

Let $\p=\{p_i\}_{i\in\sX}$ be a probability distribution with $\sX$ a finite outcome set whose label are denoted by $i$. For a given $N$,  consider the string ${\bf i}\coloneqq i_1\dots i_N$ associated with the factorized probability $p_{\bi}\coloneqq p_{i_1}\cdots p_{i_N}$. For any $\delta>0$, the ($N,\delta$)\textit{-typical set} $T^N_{\delta}(\p)$ is a subset of all the possible strings of lentgh $N$ and it is defined as follows
\[
T^N_{\delta}(\p)=\left\{ {\bf i} : \left| \frac{1}{N}\log_2\frac{1}{p_{\bf i}} - H(\p) \right| \leq \delta \right\},
\] 
where $H(\p)\coloneqq-\sum_{i\in\sX}p_i\log_2p_i$ denotes the Shannon entropy of the probability distribution $\p$. Correspondingly, a string belonging to $T^N_{\delta}(\p)$ is called a \textit{typical string}. Notice that by reformulating the definition of typical set, a string is typical iff the associated probability is bounded as follows
\begin{equation}
2^{-N[H(\p)+\delta]}\leq p_\sti \leq 2^{-N[H(\p)-\delta]} \label{eq:equip}
\end{equation}
The typical set has the following relevant properties, that will be extensively used in the proof of theorem \ref{thm:main_thm}
\begin{theorem}\label{thm:typprop}
Let $\p$ be a probability distribution, and for $\delta>0$ let $T_{\delta}^N(\p)$ denote the ($N,\delta$)-typical set. Then
\begin{enumerate}
\item \label{en:typ1} Let $\eta>0$. Then there exists $N_0$ such that for any $N\geq N_0$
\[
P({\bf i}\in T_{\delta}^N(\p))\geq 1-\eta;
\]
\item  \label{en:typ2} Let $\eta>0$. Then there exists $N_0$ such that, for any $N\geq N_0$, the cardinality of $T_{\delta}^N(\p)$ is bounded as follows
\[
(1-\eta)2^{N[H(\p)-\delta]}\leq|T^N_{\delta}(\p)|\leq 2^{N[H(\p)+\delta]}.
\] 
\end{enumerate}
\end{theorem}
This is a standard result, whose proof can be found in any of the following references \cite{nielsen_chuang_2010,wilde_2013}. With this tool in our hand, we can prove the following 

\begin{lemma} \label{lem:typlemma}
 Let $S(N)$ be any collection of strings $(\sti,\s)$, with $\sti\in\{1,\dots,D_{\rA}\}^N$  and $\s\in\{+,-\}^{N-1}$, such that $|S(N)|=2^{2NR-1}$ with $R<\frac{H(\p)+1}{2}$ fixed. Then, for any $\eta>0$ there exists $N_0$ such that 
for any $N\geq N_0$ one has
\[
\sum_{(\bi,\s)\in S(N)}p_{\bi}\frac{1}{2^{N-1}}<\eta
\]
\end{lemma}
\begin{proof}
Let $S$ be the set of all the possible strings $\s$, $\Delta \coloneqq H(\p)+1-2R>0$ and define
\begin{align*}
S_1&\coloneqq(T_{\frac{\Delta}{2}}^N(\p)\times S)\cap S(N), \\
S_2&\coloneqq(T_{\frac{\Delta}{2}}^N(\p)^c\times S)\cap S(N),
\end{align*}
where $A^c$ denotes the complementary set of $A$. Then consider
\[
\sum_{(\bi,\s)\in S(N)}p_{\bi}\frac{1}{2^{N-1}}=\sum_{(\bi,\s)\in S_1}p_{\bi}\frac{1}{2^{N-1}} + \sum_{(\bi,\s)\in S_2}p_{\bi}\frac{1}{2^{N-1}}.
\]
The first term is bounded as follows thanks to the definition of typical set (see equation \eqref{eq:equip})
\[
\begin{aligned}
\sum_{(\bi,\s)\in S_1}p_{\bi}\frac{1}{2^{N-1}}&\leq 2^{-N[H(\p)-\frac{\Delta}{2}]-N+1}|S(N)| \\
	&=2^{-N[H(\p)-\frac{\Delta}{2}]-N+1+2NR-1} \\
	&=2^{-N[H(\p)+1-2R-\frac{\Delta}{2}]} < \eta/2
\end{aligned}
\]
provided that we take $N\geq N_{1}$ with $N_1$ sufficiently large. For the term with $S_2$ we use item 1 of theorem \ref{thm:typprop}, which implies that for $\eta/2>0$ there exists $N_2$ such that for any $N\geq N_2$ we have $\sum_{\bi\in T_{\frac{\Delta}{2}}^N(\p)^c}p_\bi<\frac{\eta}{2}$. We then have
\begin{align*}
\sum_{(\bi,\s)\in S_2}p_{\bi}\frac{1}{2^{N-1}}\leq\sum_{\bi\in T_{\frac{\Delta}{2}}^N(\p)^c}p_\bi<\frac{\eta}{2}.
\end{align*}
Setting $N_0=\max\{N_1,N_2\}$ we have the thesis.
\end{proof}

Before giving the proof we provide an explicit formula to compute the figure of merit in equation \eqref{eq:BCTfom} in terms
of the probability distributions that characterise the encoding and decoding channels. Let us consider a BCT source represented by a
state $\rket{\rho}_\rA=\sum_{i\in\sX}p_i\rket{i}_\rA$ on a system $\rA$, where $\{p_i\}_{i\in\sX}$ is a probability distribution. Fix the number $N$ of copies of $\rA$ and consider a compression scheme $(\tE,\tD)$ for $M$ bibits, i.e. $\tE\in\tr_1(\rA^{\boxtimes N}\to\rB^{\boxtimes M})$ and $\tD\in\tr_1(\rB^{\boxtimes M}\to\rA^{\boxtimes N})$. By proposition \ref{prop:BCTchannels}, $\tE$ and $\tD$ are characterised by two probability distributions, $\{\lambda^{(\is)}_{\ts\tau}\}$ and $\{\mu^{(\ts)}_{\is\tau}\}$ respectively, where $\is$, $\ts$ is a shortened notation for the pairs $(\sti,{\bf s})\in\{1,\dots,D_\rA\}^N\times\{+,-\}^{N-1}$. and $({\bf t},{\bf n})\in\{1,2\}^{M}\times\{+,-\}^{M-1}$ respectively. Now, recall that the figure of merit is computed according to proposition \ref{prop:fig_merit}
\begin{multline}\label{eq:BCTfom}
\tilde{D}(\rho^{\boxtimes N},\tD\tE)= \\ 
\sum_{\is}\frac{1}{2^{N-1}}p_{\bf{i}}\normop{(\tD\tE-\tI_{\rA^N})\boxtimes\tI_{\rA^N})|(\is\is)_+)_{\rA^N\rA^N}}.
\end{multline}
For any $\is$ consider the norm
\begin{equation*}
\begin{aligned}
\normop{&(\tD\tE-\tI_{\rA^N})\boxtimes\tI_{\rA^N})|(\is\is)_+)_{\rA^N\rA^N}} \\
&=\normop{\sum_{\ts,{\bf i'_{s'}},\tau,\tau'} \lambda^{(\is)}_{\ts\tau}\mu^{(\ts)}_{{\bf i'_{s'}}\tau'}\rket{({\bf i'_{s'}}\is)_{\tau\tau'}}-\rket{(\is\is)_{+}}} \\
&= \normop{\sum_{\ts,{\bf i'_{s'}}\neq\is,\tau,\tau'}\lambda^{(\is)}_{\ts\tau}\mu^{(\ts)}_{{\bf i'_{s'}}\tau'}\rket{({\bf i'_{s'}}\is)_{\tau\tau'}}
	\\ &\quad+\sum_{\ts,\tau\neq\tau'}\lambda^{(\is)}_{\ts\tau}\mu^{(\ts)}_{\is\tau'}\rket{(\is\is)_-}
	\\ &\quad \left(-1+ \sum_{\ts,\tau}\lambda^{(\is)}_{\ts\tau}\mu^{(\ts)}_{{\is},\tau} \right) \rket{(\is\is)_{+}}}
\end{aligned}
\end{equation*}
where we have decomposed the sum into three terms.
For any element $\rho$ in 
$\st_\R(\rA)$, say $\rho=\sum_{i}c_i\rket{i}_\rA$, we have $\normop{\rho}=\sum_i|c_i|$. Thereofre, using the normalisation conditions
for $\{\lambda^{(\is)}_{\ts\tau}\}$ and $\{\mu^{(\ts)}_{\is\tau}\}$
\[
\begin{aligned}
\normop{&(\tD\tE-\tI_{\rA^N})\boxtimes\tI_{\rA^N})|(\is\is)_+)_{\rA^N\rA^N}} \\
&= 2\left(1 - \sum_{\ts,\tau}\lambda^{(\is)}_{\ts\tau}\mu^{(\ts)}_{{\is},\tau}\right)
\end{aligned}
\]
and the figure of merit then becomes 
\begin{equation}\label{eq:BCTfomMod}
\tilde{D}(\rho^{\boxtimes N},\tD\tE)=2\left(1-\sum_{\is}\frac{1}{2^{N-1}}p_{\sti}\sum_{\ts,\tau}\lambda_{\ts\tau}^{(\is)}\mu^{(\ts)}_{\is\tau}\right).
\end{equation}

\begin{proof}[\bf Proof of theorem \ref{thm:main_thm}]
We first prove the achievability, namely that $I(\rho)\leq\frac{H(\{p_i\})+1}{2}$. Let $\delta>0$, and for any $N$ consider the following number of bibits

\[
M=\left\lceil N\left[\frac{H(\p)+1}{2}+\delta\right] \right\rceil
\]
by item \ref{en:typ2} of theorem \ref{thm:typprop} and the above choice of $M$ we have that, $\forall N \in \N$ 
\[
|T_{\delta}^N(\p)|2^{N-1}\leq 2^{N[H(\p)+1+\delta]-1}\leq 2^{2M-1} = D_{\rB^{\boxtimes M}}.
\] 
This bound entails the existence of a subset $\cP_N\subseteq\purst(\rB^{\boxtimes M})$ with cardinality equal to $|T_{\delta}^N(\p)|2^{N-1}$. Denote by $S_{\cP_N}$ the set of strings associated with $\cP_N$. We will use the same notation that we introduced in the proof of lemma \ref{lem:BCTdigitizability} and below the proof of lemma \ref{lem:typlemma}: any element in $S_{\cP_N}$ is specified by $M$ signs ${\bf t}=t_1,\dots,t_M$, one for each $\rB$ in the composition, and $M-1$ additional signs ${\bf n}=n_1,\dots,n_{M-1}$ emerging from the violation of purity of parallel composition. We then denote any element of $S_{\cP_N}$ with the multi-index $\bf t_n$. We now define the following $(\tE,\tD)$ compression scheme for any $N$:
\begin{enumerate}
\item The encoding $\tE$ is defined by a set of probability distributions $\lambda^{(\is)}_{\ts\tau}$, one for each multi-index $\is$, where ${\bf i}\in \{1,\dots,D_{\rA}\}^N$ and ${\bf s}\in\{+,-\}^{N-1}$ (recall the notation in the proof of lemma \ref{lem:BCTdigitizability}). Let $h:T^N_{\delta}(\p)\times \{+,-\}^{N-1}\rightarrow S_{\cP_N}$ be any injective function that associates each typical string $\is$ with a distinct element $\ts$ of $S_{\cP_N}$. Then, for any such $\is$ we set $\lambda_{h(\is)+}^{(\is)}=1$. In particular, for any $\sti\in T_{\delta}^N(\p)$ and any $\s \in \{+,-\}^{N-1}$ we have the following diagrammatic equation 

\begin{align*}
\begin{aligned}
\Qcircuit @C=1.2em @R=1em @! R { \multiprepareC{1}{(\is\is)_+} & \ustick {\rA^{\boxtimes N}}  \qw  & \gate {\tE} & \ustick {\rB^{\boxtimes M}}  \qw  &\qw \\
\pureghost {(\is\is)_+} & \qw  &\ustick  {\rA^{\boxtimes N}} \qw &\qw &\qw }
\end{aligned}
\ = \
\begin{aligned}
\Qcircuit @C=1.8em @R=1em @! R { \multiprepareC{1}{(h(\is)\is)_+} & \ustick {\rB^{\boxtimes M}} \qw & \qw \\
\pureghost {(h(\is)\is)_+} & \ustick {\rA^{\boxtimes N}} \qw & \qw}
\end{aligned}\ .
\end{align*}
If $\sti\not\in T^N_{\delta}(\p)$, for any $\s$ set $\lambda_{\bf \bar{t}_{\bar{s}}+}^{(\is)}=1$ for a fixed ${\bf \bar{t}_{\bar{s}}}\in S_{\cP_N}$.
Diagrammatically the action of $\tE$ is represented as follows
\begin{align*}
\begin{aligned}
\Qcircuit @C=1.2em @R=1em @! R { \multiprepareC{1}{(\is\is)_+} & \ustick {\rA^{\boxtimes N}}  \qw  & \gate {\tE} & \ustick {\rB^{\boxtimes M}}  \qw  &\qw \\
\pureghost {(\is\is)_+} & \qw  &\ustick  {\rA^{\boxtimes N}} \qw &\qw &\qw }
\end{aligned}
\ = \
\begin{aligned}
\Qcircuit @C=1.8em @R=1em @! R { \multiprepareC{1}{(\bf \bar{t}_{\bar{s}}\is)_+} & \ustick {\rB^{\boxtimes M}} \qw & \qw \\
\pureghost {(\bf \bar{t}_{\bar{s}}\is)_+} & \ustick {\rA^{\boxtimes N}} \qw & \qw}
\end{aligned}\ .
\end{align*}
Notice that, having defined $\lambda^{(\is)}_{\ts\tau}$ for any $\ts$ and $\tau$ and for any $\is$, we have fully specified the action of $\tE$ on all the pure states of $\rA^{\boxtimes N}\boxtimes \rA^{\boxtimes N}$ even if not all of them directly intervene in the evaluation of the figure of merit. 
\item Let $\{\mu_{\is\tau}^{(\ts)}\}$ be the probability distributions defining the decoding $\tD$. For any  $\ts\in S_{\cP_N}$ we simply set $\mu_{h^{-1}(\ts)+}^{(\ts)}=1$, namely, we invert the action of the encoding. Indeed, for any typical string $\is=h^{-1}(\ts)$ we have

\begin{align*}
\begin{aligned}
\Qcircuit @C=1.2em @R=1em @! R { \multiprepareC{1}{(\ts\is)_+} & \ustick {\rB^{\boxtimes M}}  \qw  & \gate {\tD} & \ustick {\rA^{\boxtimes N}}  \qw  &\qw \\
\pureghost {(\ts\is)_+} & \qw  &\ustick  {\rA^{\boxtimes N}} \qw &\qw &\qw }
\end{aligned}
\ = \
\begin{aligned}
\Qcircuit @C=1.8em @R=1em @! R { \multiprepareC{1}{(h^{-1}(\ts)\is)_+} & \ustick {\rA^{\boxtimes N}} \qw & \qw \\
\pureghost {(h^{-1}(\ts)\is)_+} & \ustick {\rA^{\boxtimes N}} \qw & \qw}
\end{aligned}\ .
\end{align*}
If $\ts\not\in S_{\cP_N}$, take a fixed string $\bar{\sti}_{\bar{\s}}$ and define $\mu_{\bar{\sti}_{\bar{\s}}+}^{(\ts)}=1$. This implies, for any $\is$

\begin{align*}
\begin{aligned}
\Qcircuit @C=1.2em @R=1em @! R { \multiprepareC{1}{(\ts\is)_+} & \ustick {\rB^{\boxtimes M}}  \qw  & \gate {\tD} & \ustick {\rA^{\boxtimes N}}  \qw  &\qw \\
\pureghost {(\ts\is)_+} & \qw  &\ustick  {\rA^{\boxtimes N}} \qw &\qw &\qw }
\end{aligned}
\ = \
\begin{aligned}
\Qcircuit @C=1.8em @R=1em @! R { \multiprepareC{1}{(\bar{\sti}_{\bar{\s}}\is)_+} & \ustick {\rA^{\boxtimes N}} \qw & \qw \\
\pureghost {(\bar{\sti}_{\bar{\s}}\is)_+} & \ustick {\rA^{\boxtimes N}} \qw & \qw}
\end{aligned}\ .
\end{align*}
\end{enumerate}
Now, with this scheme and using item \ref{en:typ1} of theorem \ref{thm:typprop}, we have that for any $\eta>0$ there exists $N_0$ such that for any $N\geq N_0$ the following holds
\begin{align*}
&D(\rho^{\boxtimes N},\tC)= \\ 
&\sum_{\bf{i}\not\in T_{\delta}^{N}(\p),\s}\frac{1}{2^{N-1}}p_{\bf{i}}\normop{|(\bar{\sti}_{\bar{\s}}\is)_+)_{\rA^N\rA^N}-|\is\is)_+)_{A^NA^N}} \\
&\leq 2\sum_{\bf{i}\not\in T_{\delta}^{N}(\p),\s}\frac{1}{2^{N-1}}p_{\bf{i}}\leq2\sum_{\bf{i}\not\in T_{\delta}^{N}(\p)}p_{\bf{i}} \leq 2\eta
\end{align*}
which in turns implies that, given $\epsilon=2\eta$, for any $N\geq N_0$ holds true that $E_{N,M,\epsilon}(\rho)\neq \emptyset$. Therefore, for such values of $N$ we have $M_{N,\epsilon}(\rho)/N\leq M/N$ (recall definition \ref{eq:IC-def}), and this implies, by taking the $\limsup_{N\rightarrow \infty}$ and then the $\lim_{\epsilon\rightarrow0}$ 
\[
I(\rho)\leq \frac{H(\p)+1}{2}+\delta,
\]
and the thesis follows by the arbitrariness of $\delta$.

Now we prove the minimality of $\frac{H(\p)+1}{2}$, namely that $I(\rho)\geq \frac{H(\p)+1}{2}$. Fix an arbitrary real positive number $\delta>0$ and let $\overline{M}$ and $\overline{N}$ be such that
\begin{equation}\label{eq:MbarNbar}
\frac{H(\p)+1}{2}-\delta\leq\frac{\overline M}{\overline {N}}<\frac{H(\p)+1}{2}.
\end{equation}
We then show that there exists $N_0$ such that for any integer $k$ satisfying
$k\overline{N}>N_0$ it holds that $E_{k\overline{N},k\overline{M},\epsilon}(\rho)=\emptyset$ for any
$\epsilon\in(0,\overline{\epsilon}]$ for some $\overline{\epsilon}$.
Let $(\tE,\tD)$ be a compression scheme for messages of length $k\overline{N}$ with $k\overline{M}$ bibits and let $\lambda_{\ts\tau}^{(\is)}$, $\mu^{(\ts)}_{\is\tau}$ be the probability distributions defining the action of $\tE$ and $\tD$ respectively on pure states of $\rA^{\boxtimes k\overline{N}}\rA^{\boxtimes k\overline{N}}$. Recall the expression of the figure of merit given in equation \eqref{eq:BCTfomMod}:
\begin{equation}
D(\rho^{\boxtimes k\overline{N}},\tD\tE)=2\left(1-\sum_{\is}\frac{1}{2^{N-1}}p_{\sti}\sum_{\ts,\tau}\lambda_{\ts\tau}^{(\is)}\mu^{(\ts)}_{\is\tau}\right).
\end{equation}
Now consider the sum between brackets. Since all the terms $\lambda_{\ts\tau}^{(\is)}\mu^{(\ts)}_{\is\tau'}$ are non-negative, we can upper bound the sum over $\tau=\tau'$ by a sum over independent indices $\tau,\tau'$, obtaining
\begin{equation}\label{eq:stupidbound}
\sum_{\is}\frac{1}{2^{N-1}}p_{\sti}\sum_{\ts,\tau}\lambda_{\ts\tau}^{(\is)}\mu^{(\ts)}_{\is\tau} \leq
\sum_{\is}\frac{1}{2^{N-1}}p_{\sti}\sum_{\ts}\lambda_{\ts}^{(\is)}\mu^{(\ts)}_{\is}
\end{equation}
where $\lambda_{\ts}^{(\is)}\coloneqq\sum_{\tau}\lambda_{\ts\tau}^{(\is)}$ and similarly for $\mu^{(\ts)}_{\is}$. Now we check that the right hand side of the inequality
written above is bounded as follows (we set $p_{\is}=p_{\sti}\frac{1}{2^{N-1}}$)
\begin{equation}
\sum_{\is}p_{\is}\sum_{\ts}\lambda_{\ts}^{(\is)}\mu^{(\ts)}_{\is}\leq \sum_{\is\in \mathsf{I}}p_{\is},
\end{equation}
where $\mathsf{I}$ is a set with cardinality at most $2^{2k\overline{M}-1}$; indeed, by repeatedly using the inequality $\sum_ka_kb_k\leq\max_k\{b_k\}$ valid for $b_k\geq0$ and $\sum_ka_k=1$ we find the following bound
\begin{equation}\label{eq:stupidbound2}
\begin{aligned}
\sum_{\is}p_{\is}\sum_{\ts}&\lambda_{\ts}^{(\is)}\mu^{(\ts)}_{\is} \leq\sum_{\is}p_{\is}\mu_{\is}^{(\tnbar(\is))} \\
		& =\sum_{\ts}\sum_{\is:\tnbar(\is)=\ts}p_{\is}\mu_{\is}^{(\tnbar(\is))} \\
		& \leq \sum_{\ts}p_{\is(\ts)}
\end{aligned}
\end{equation}
where, for any $\is$, $\tnbar(\is)$ is a fixed multi-index that corresponds to the maximum value $\mu_{\is}^{(\tnbar(\is))}$. In the last sum we have only one $\is$ for any $\ts$, therefore, since the cardinality of the set of strings $\ts$ is $2^{2k\overline{M}-1}$, the sum is over at most $2^{2k\overline{M}-1}$ terms, as claimed. 

Now, notice that our choice of $M_k=k\overline{M}$ entails 
$2^{2M_k-1}= 2^{2\frac{\overline{M}}{\overline{N}}N_k-1}$.
Since $\frac{\overline{M}}{\overline{N}}<\frac{H(\p)+1}{2}$ (see equation \eqref{eq:MbarNbar}), lemma \ref{lem:typlemma} can be applied, so that given $\eta>0$ there exists $N_0$ such that for any $k$ satisfying $k\overline{N}>N_0$ we have $\sum_{\ts}p_{\is(\ts)}<\eta$. By \eqref{eq:stupidbound} and \eqref{eq:stupidbound2}, we have (recall again \eqref{eq:BCTfomMod})
\[
\begin{aligned}
D(\rho^{\boxtimes k\overline{N}},\tD\tE)&\geq 2\left(1 - \sum_{\is}\frac{1}{2^{N-1}}p_{\sti}\sum_{\ts}\lambda_{\ts}^{(\is)}\mu^{(\ts)}_{\is}\right) \\
	&\geq 2\left(1 - \sum_{\ts}p_{\is(\ts)}\right) >2(1- \eta).
\end{aligned}
\]
And this holds for any compression scheme $(\tE,\tD)$ for $k\overline{N}$ usages of the source with $k\overline{M}$ bibits. This implies that, for any $\epsilon\in(0,2(1-\eta)]$
\[
E_{k\overline{N},k\overline{M},\epsilon}=\emptyset, \qquad \forall k>\left\lceil\frac{N_0}{\overline{N}}\right\rceil.
\]
Therefore, for any such $k$, we have the following chain of inequalities
\begin{align*}
\frac{H(\p)+1}{2}-\delta\leq \frac{k\overline{M}}{k\overline{N}}& \leq\limsup_{k\rightarrow \infty}\frac{M_{k\overline{N},\epsilon}(\rho)}{k\overline{N}} \\
&\leq\limsup_{N\rightarrow \infty}\frac{M_{N,\epsilon}(\rho)}{N}.
\end{align*}
The first and the second inequality follow by equation \eqref{eq:MbarNbar} and the definition of $M_{k\overline{N},\varepsilon}(\rho)$ respectively, while the third inequality follows by the fact that the $\limsup$ of a real sequence is always greater than the $\limsup$ of any of its subsequence. Taking the $\lim_{\epsilon\rightarrow 0}$ we have $\frac{H(\p)+1}{2}-\delta\leq I(\rho)$ and the thesis follows by the arbitrariness of $\delta$.
\end{proof}

Before giving the proof of proposition \ref{prop:regentr_noatomicity} we recall the following lemma (lemma 1 in \cite{PhysRevA.101.042118}).
\begin{lemma}\label{lem:lemma1}
Consider a simplicial OPT satisfying $n$-local discriminability for some positive integer $n$. Then for all systems $\rA$, $\rB$ and non-null extremal states 
$\rket{k}_{\rA\rB}$ of the composite system $\rA\rB$, there exists a unique product of non-null extremal states $\rket{i_k}_{\rA}\boxtimes\rket{j_k}_\rB$ such that $\rket{k}_{\rA\rB}$ convexly refines $\rket{i_k}_{\rA}\boxtimes\rket{j_k}_\rB$.
\end{lemma}
Since we are focusing on strongly causal theories, pure states are all and only the non-null extremal states, as already mentioned at the end of subsection \ref{subsec:framework}. 
The above lemma tells us that for any $\rket{k}_{\rA\rB}\in\purst(\rA\rB)$ there exist $\rket{i_k}_{\rA}\in\purst(\rA)$, $\rket{j_k}_\rB\in\purst(\rB)$ and  a probability distribution $\{p_\ell\}_{\ell\in\sL}$  such that 
\begin{equation}
\rket{i_k}_\rA\boxtimes\rket{j_k}_\rB=\sum_{\ell\in\sL}p_\ell \rket{\ell}_{\rA\rB}
\end{equation}
and $\rket{k}_{\rA\rB}=\rket{\ell}_{\rA\rB}$ for some $\ell\in\sL$. 
It easy to see that for any pure state $\rket{k}_{\rA\rB}$ of a bipartite system $\rA\rB$, both
$\rbra{j_k}_\rB\rket{k}_{\rA\rB}$ and $\rbra{i_k}_\rA\rket{k}_{\rA\rB}$ are also pure.
Indeed, if we now apply the effect $\rbra{j_k}_\rB$ on $\rB$ to the equation above, by purity of $\rket{i_k}$ it follows $\rket{i_k}_\rA=\rbra{j_k}_\rB\rket{\ell}_{\rA\rB}$ for any $\ell\in\sL$, and in particular for $\rket{k}$, which is then pure. Similarly if one applies the effect $\rket{i_k}$ on $\rA$. It will be useful to keep in mind this fact in the proof of proposition \ref{prop:regentr_noatomicity}
 
\begin{proof}[\bf Proof of proposition \ref{prop:regentr_noatomicity}] 
By hypothesis, there exist systems $\rA$ and $\rB$ and $\rket{i}_{\rA}\in\purst(\rA)$, $\rket{j}_{\rB}\in\purst(\rB)$ such that 

\begin{align*}
\begin{aligned}
\Qcircuit @C=1em @R=.7em @! R { \multiprepareC{1}{\Sigma} & \ustick  \rA \qw &\qw \\
\pureghost {\Sigma}&  \ustick \rB  \qw   &  \qw }
\end{aligned}
\ \coloneqq \
\begin{aligned}
\Qcircuit @C=1em @R=.7em @! R { \prepareC{i} &\ustick  \rA \qw  &\qw \\
\prepareC {j} & \ustick \rB  \qw   &\qw }
\end{aligned}
\ =\sum_{k}p_k^{(ij)} \
\begin{aligned}
\Qcircuit @C=1em @R=.7em @! R { \multiprepareC{1}{k} & \ustick \rA \qw & \qw \\
\pureghost {k} & \ustick \rB \qw & \qw}
\end{aligned} 
\end{align*}
where $\{p_k^{(ij)}\}_{k\in I_{ij}}$ is a non-trivial probability distribution and $\{\rket{k}_{\rA\rB}\}_{k\in I_{ij}}$ is a set of pure states of $\rA\rB$. 
Now consider $\rket{\Sigma^{\boxtimes 2}}_{\rA\rB}$, that can be decomposed as follows

\begin{align*}
\begin{aligned}
\Qcircuit @C=1em @R=.7em @! R { \multiprepareC{1}{\Sigma} & \ustick  \rA \qw &\qw \\
							\pureghost {\Sigma}&  \ustick \rB  \qw   &  \qw \\
							 \multiprepareC{1}{\Sigma} & \ustick  \rA \qw &\qw \\
							\pureghost {\Sigma}&  \ustick \rB  \qw   &  \qw}
\end{aligned}
\ = \sum_{k,k'} p_k p_{k'} \
\begin{aligned}
\Qcircuit @C=1em @R=.7em @! R { \multiprepareC{1}{k} & \ustick  \rA \qw &\qw \\
							\pureghost {k}&  \ustick \rB  \qw   &  \qw \\
							 \multiprepareC{1}{k'} & \ustick  \rA \qw &\qw \\
							\pureghost {k'}&  \ustick \rB  \qw   &  \qw}.
\end{aligned}
\end{align*}
If $\rket{k}_{\rA\rB}\boxtimes\rket{k'}_{\rA\rB}$ is pure, then one has a contradiction in that
\begin{align*}
\begin{aligned}
\Qcircuit @C=1em @R=.7em @! R { \prepareC{j} & \ustick  \rB \qw &\qw \\
							\prepareC{i}&  \ustick \rA  \qw   &  \qw}
\end{aligned}
\ = \
\begin{aligned}
\Qcircuit @C=1em @R=.7em @! R { \multiprepareC{1}{k} & \ustick  \rA \qw &\measureD{i} \\
							\pureghost {k}&  \ustick \rB  \qw   &  \qw \\
							 \multiprepareC{1}{k'} & \ustick  \rA \qw &\qw \\
							\pureghost {k'}&  \ustick \rB  \qw   &  \measureD{j}}
\end{aligned}.
\end{align*}
and, while the left hand side is mixed by hypothesis, the right hand side should be pure, since the whole state is assumed to be pure (see the discussion below lemma \ref{lem:lemma1}). We then conclude that $\rket{k}_{\rA\rB}\boxtimes\rket{k'}_{\rA\rB}$ must be necessarily mixed. Therefore, $\Sigma^{\boxtimes 2}$ has the following decomposition in terms of pure states
\begin{align*}
\begin{aligned}
\Qcircuit @C=1em @R=.7em @! R { \multiprepareC{1}{\Sigma} & \ustick  \rA \qw &\qw \\
							\pureghost {\Sigma}&  \ustick \rB  \qw   &  \qw \\
							 \multiprepareC{1}{\Sigma} & \ustick  \rA \qw &\qw \\
							\pureghost {\Sigma}&  \ustick \rB  \qw   &  \qw}
\end{aligned}
\ = \sum_{k,k'} p_k p_{k'} q_{\ell|k,k'} \
\begin{aligned}
\Qcircuit @C=1em @R=.7em @! R { \multiprepareC{3}{\ell} & \ustick  \rA \qw &\qw \\
							\pureghost {\ell}&  \ustick \rB  \qw   &  \qw \\
							 \pureghost{\ell} & \ustick  \rA \qw &\qw \\
							\pureghost {\ell}&  \ustick \rB  \qw   &  \qw}.
\end{aligned}
\end{align*}
where at least one of the conditioned probability distributions $\{q_{\ell|k,k'}\}_{\ell}$ is non-trivial. Now, notice that for any system $\rA$ and any state
$\rket{\rho}_{\rA}=\sum_ip_i\rket{i}_{\rA}$ one has $S_i(\rho)=H(\{p\})$. The argument that we give here is the same as the one proposed in \cite{KIMURA2010175} for classical theory (see in particular proposition 13 and theorem 3(i)), but since states are separating for effects in any OPT it also applies to the case of \emph{any} classical theory defined according to definition \ref{def:classicaltheories}. The case of $S_3$ is pretty obvious, given the uniqueness of the decomposition in terms of pure states. Now, since states are separating for effects, any other atomic effect is proportional to an effect of the perfectly discriminating test, i.e., $\rbra{a}_{\rA}=\lambda\rbra{i}_{\rA}$ for some $i$. Thus, for any state $\rho$ of a classical theory, it follows that any other atomic observation test $\{a_i\}$ is such that $H(a_i(\rho))\geq H(\p)$, by concavity of the function $x\log x$, and the equality is achieved for the perfectly discriminating test, whence $S_1(\rho)=H(\p)$. Finally, notice that for any state $\rho$ of a classical theory one has 
\[
S_2(\rho)=H(\p)-\inf_{\{a_i\}\in\cO^{\rm at}}H(X|J)
\]
and $H(X|J)=0$ for the perfectly discriminating test.
Now, a trivial computation shows that
\[
\begin{aligned}
S_i(\Sigma^{\boxtimes 2})= \ &-\sum_{k,k',\ell}p_kp_{k'}q_{\ell|k,k'}\log_2(p_kp_{k'}q_{\ell|k,k'})= \\
	= \ & 2S_i(\Sigma) -\sum_{k,k'\ell}p_kp_{k'}q_{\ell|k,k'}\log_2(q_{\ell|k,k'}) \\
	> \ & 2S_i(\Sigma)
\end{aligned}
\]  
Finally, notice that the subsequence $S_i(\Sigma^{\boxtimes 2^k})/2^k$ is increasing, since
\[
\frac{S_i(\Sigma^{\boxtimes 2^k})}{2^k}\geq \frac{S_i(\Sigma^{\boxtimes 2^{k-1}})}{2^k} + \frac{S_i(\Sigma^{\boxtimes 2^{k-1}})}{2^k} = \frac{S_i(\Sigma^{\boxtimes 2^{k-1}})}{2^{k-1}}
\]
and the result follows since the $\limsup$ of the whole sequence is greater than the $\limsup$ of any of its subequence
\[
\limsup_{N\to\infty}\frac{S_i(\Sigma^{\boxtimes N})}{N}\geq\lim_{k\to\infty}\frac{S_i(\Sigma^{\boxtimes 2^k})}{2^k}\geq\frac{S_i(\Sigma^{\boxtimes 2})}{2}> S_i(\Sigma)
\].
\end{proof}
\begin{proof}[\bf Proof of proposition \ref{prop:entropiesBCT}]
BCT is a classical theory, therefore we immediately have that $S_i(\rho)=H(\p)$ for any $i=1,2,3$ (see proposition \ref{prop:regentr_noatomicity}). For $S^{\rm reg}_i$, we just notice that $S_i(\rho^{\boxtimes N})$ is the Shannon of the factorized joint distribution $p_{\sti,\s}=p_{\bi}\frac{1}{2^{N-1}}$ where $p_{\sti}=p_{i_1}\dots p_{i_N}$, thus 
\[
\frac{S_i(\rho^{\boxtimes N})}{N}=\frac{NH(\p)+(N-1)H(\frac{1}{2})}{N}=H(\p)+1-\frac{1}{N}
\]
and the result follows by taking the limit.
\end{proof}
\end{document}

%% file: sample.tikzstyles
\tikzstyle{new style 0}=[fill=white, draw=black, shape=circle]
\tikzstyle{rectangle}=[fill=white, draw=black, shape=rectangle]

\tikzstyle{new edge style 0}=[-, thick]
\tikzstyle{new edge style 1}=[-, dotted, thick]
\tikzstyle{new edge style 2}=[->]

%% file: myqcircuit.tex
\usepackage[matrix,frame,arrow]{xy}
\usepackage{amsmath}

\newcommand{\qw}[1][-1]{\ar @{-} [0,#1]}

\newcommand{\gate}[1]{*{\xy *+<.6em>{#1};p\save+LU;+RU **\dir{-}\restore\save+RU;+RD **\dir{-}\restore\save+RD;+LD **\dir{-}\restore\POS+LD;+LU **\dir{-}\endxy} \qw}

\newcommand{\measureD}[1]{*{\xy*+=+<.5em>{\vphantom{\rule{0em}{.1em}#1}}*\cir{r_l};p\save*!R{#1} \restore\save+UC;+UC-<.5em,0em>*!R{\hphantom{#1}}+L **\dir{-} \restore\save+DC;+DC-<.5em,0em>*!R{\hphantom{#1}}+L **\dir{-} \restore\POS+UC-<.5em,0em>*!R{\hphantom{#1}}+L;+DC-<.5em,0em>*!R{\hphantom{#1}}+L **\dir{-} \endxy} \qw}

\newcommand{\multimeasureD}[2]{*+<1em,.9em>{\hphantom{#2}}\save[0,0].[#1,0];p\save !C *{#2},p+LU+<0em,0em>;+RU+<-.8em,0em> **\dir{-}\restore\save +LD;+LU **\dir{-}\restore\save +LD;+RD-<.8em,0em> **\dir{-} \restore\save +RD+<0em,.8em>;+RU-<0em,.8em> **\dir{-} \restore \POS !UR*!UR{\cir<.9em>{r_d}};!DR*!DR{\cir<.9em>{d_l}}\restore \qw}

\newcommand{\multigate}[2]{*+<1em,.9em>{\hphantom{#2}} \qw \POS[0,0].[#1,0];p !C *{#2},p \save+LU;+RU **\dir{-}\restore\save+RU;+RD **\dir{-}\restore\save+RD;+LD **\dir{-}\restore\save+LD;+LU **\dir{-}\restore}
\newcommand{\ghost}[1]{*+<1em,.9em>{\hphantom{#1}} \qw}
\newcommand{\ustick}[1]{*!D!<0em,-.5em>=<0em>{#1}}

\newcommand{\Qcircuit}[1][0em]{\xymatrix @*=<#1>}
\newcommand{\node}[2][]{{\begin{array}{c} \ _{#1}\  \\ {#2} \\ \ \end{array}}\drop\frm{o} }

\newcommand{\pureghost}[1]{*+<1em,.9em>{\hphantom{#1}}}
\newcommand{\multiprepareC}[2]{*+<1em,.9em>{\hphantom{#2}}\save[0,0].[#1,0];p\save !C
  *{#2},p+RU+<0em,0em>;+LU+<+.8em,0em> **\dir{-}\restore\save +RD;+RU **\dir{-}\restore\save
  +RD;+LD+<.8em,0em> **\dir{-} \restore\save +LD+<0em,.8em>;+LU-<0em,.8em> **\dir{-} \restore \POS
  !UL*!UL{\cir<.9em>{u_r}};!DL*!DL{\cir<.9em>{l_u}}\restore}
\newcommand{\prepareC}[1]{*{\xy*+=+<.5em>{\vphantom{#1\rule{0em}{.1em}}}*\cir{l^r};p\save*!L{#1} \restore\save+UC;+UC+<.5em,0em>*!L{\hphantom{#1}}+R **\dir{-} \restore\save+DC;+DC+<.5em,0em>*!L{\hphantom{#1}}+R **\dir{-} \restore\POS+UC+<.5em,0em>*!L{\hphantom{#1}}+R;+DC+<.5em,0em>*!L{\hphantom{#1}}+R **\dir{-} \endxy}}


%% file: braiding.tikz
\begin{tikzpicture}
	\begin{pgfonlayer}{nodelayer}
		\node [style=none] (0) at (-1.25, -1) {};
		\node [style=none] (1) at (-1.25, 1) {};
		\node [style=none] (2) at (1.25, 1) {};
		\node [style=none] (3) at (1.25, -1) {};
		\node [style=none] (4) at (-0.25, -0.25) {};
		\node [style=none] (5) at (0.25, 0.25) {};
		\node [style=none] (6) at (1.25, 1) {};
		\node [style=none] (7) at (1.25, -1) {};
		\node [style=none] (10) at (-1.75, 1) {};
		\node [style=none] (11) at (-1.75, -1) {};
		\node [style=none] (12) at (-1.25, 1) {};
		\node [style=none] (13) at (-1.25, -1) {};
		\node [style=none] (16) at (1.75, 1) {};
		\node [style=none] (17) at (1.75, -1) {};
		\node [style=none] (18) at (-1.75, 1.75) {$\rA$};
		\node [style=none] (19) at (1.75, 1.75) {$\rB$};
		\node [style=none] (20) at (-1.75, -0.25) {$\rB$};
		\node [style=none] (21) at (1.75, -0.25) {$\rA$};
	\end{pgfonlayer}
	\begin{pgfonlayer}{edgelayer}
		\draw [in=-180, out=0, looseness=0.75] (1.center) to (3.center);
		\draw [in=-135, out=0] (0.center) to (4.center);
		\draw [in=180, out=45] (5.center) to (2.center);
		\draw (10.center) to (12.center);
		\draw (11.center) to (13.center);
		\draw (6.center) to (16.center);
		\draw (7.center) to (17.center);
	\end{pgfonlayer}
\end{tikzpicture}

%% file: braiding_inv.tikz
\begin{tikzpicture}
	\begin{pgfonlayer}{nodelayer}
		\node [style=none] (0) at (-1.25, -1) {};
		\node [style=none] (1) at (-1.25, 1) {};
		\node [style=none] (2) at (1.25, 1) {};
		\node [style=none] (3) at (1.25, -1) {};
		\node [style=none] (4) at (0.25, -0.25) {};
		\node [style=none] (5) at (-0.25, 0.25) {};
		\node [style=none] (6) at (1.25, 1) {};
		\node [style=none] (7) at (1.25, -1) {};
		\node [style=none] (10) at (-1.75, 1) {};
		\node [style=none] (11) at (-1.75, -1) {};
		\node [style=none] (12) at (-1.25, 1) {};
		\node [style=none] (13) at (-1.25, -1) {};
		\node [style=none] (16) at (1.75, 1) {};
		\node [style=none] (17) at (1.75, -1) {};
		\node [style=none] (18) at (-1.75, 1.75) {$\rB$};
		\node [style=none] (19) at (1.75, 1.75) {$\rA$};
		\node [style=none] (20) at (-1.75, -0.25) {$\rA$};
		\node [style=none] (21) at (1.75, -0.25) {$\rB$};
	\end{pgfonlayer}
	\begin{pgfonlayer}{edgelayer}
		\draw (10.center) to (12.center);
		\draw (11.center) to (13.center);
		\draw (6.center) to (16.center);
		\draw (7.center) to (17.center);
		\draw [in=-180, out=0, looseness=0.75] (13.center) to (6.center);
		\draw [in=135, out=0] (12.center) to (5.center);
		\draw [in=-180, out=-45] (4.center) to (7.center);
	\end{pgfonlayer}
\end{tikzpicture}

%% file: swap_LHS.tikz
\begin{tikzpicture}
	\begin{pgfonlayer}{nodelayer}
		\node [style=none] (0) at (-1.25, -1) {};
		\node [style=none] (1) at (-1.25, 1) {};
		\node [style=none] (2) at (1.25, 1) {};
		\node [style=none] (3) at (1.25, -1) {};
		\node [style=none] (4) at (-0.25, -0.25) {};
		\node [style=none] (5) at (0.25, 0.25) {};
		\node [style=none] (6) at (1.25, 1) {};
		\node [style=none] (7) at (1.25, -1) {};
		\node [style=none] (8) at (-1.25, 1) {};
		\node [style=none] (9) at (-1.25, -1) {};
		\node [style=none] (10) at (1.25, 1) {};
		\node [style=none] (11) at (1.25, -1) {};
		\node [style=none] (12) at (1.75, 1) {};
		\node [style=none] (13) at (1.75, -1) {};
		\node [style=none] (16) at (-1.25, -1) {};
		\node [style=none] (17) at (-1.25, 1) {};
		\node [style=rectangle] (18) at (-3, 1) {$\tA$};
		\node [style=rectangle] (19) at (-3, -1) {$\tB$};
		\node [style=none] (24) at (-4.25, 1) {};
		\node [style=none] (25) at (-4.25, -1) {};
		\node [style=none] (26) at (-4.25, 1.75) {$\rA$};
		\node [style=none] (27) at (-1.75, 1.75) {$\rB$};
		\node [style=none] (28) at (-4.25, -0.25) {$\rC$};
		\node [style=none] (29) at (-1.75, -0.25) {$\rD$};
		\node [style=none] (30) at (1.75, -0.25) {$\rB$};
		\node [style=none] (31) at (1.75, 1.75) {$\rD$};
	\end{pgfonlayer}
	\begin{pgfonlayer}{edgelayer}
		\draw [in=-180, out=0, looseness=0.75] (1.center) to (3.center);
		\draw [in=-135, out=0] (0.center) to (4.center);
		\draw [in=180, out=45] (5.center) to (2.center);
		\draw (10.center) to (12.center);
		\draw (11.center) to (13.center);
		\draw (18) to (17.center);
		\draw (19) to (16.center);
		\draw (24.center) to (18);
		\draw (25.center) to (19);
	\end{pgfonlayer}
\end{tikzpicture}

%% file: swap_RHS.tikz
\begin{tikzpicture}
	\begin{pgfonlayer}{nodelayer}
		\node [style=none] (0) at (-1.25, -1) {};
		\node [style=none] (1) at (-1.25, 1) {};
		\node [style=none] (2) at (1.25, 1) {};
		\node [style=none] (3) at (1.25, -1) {};
		\node [style=none] (4) at (-0.25, -0.25) {};
		\node [style=none] (5) at (0.25, 0.25) {};
		\node [style=none] (6) at (4.25, -1) {};
		\node [style=none] (7) at (4.25, 1) {};
		\node [style=rectangle] (8) at (3, 1) {$\tB$};
		\node [style=rectangle] (9) at (3, -1) {$\tA$};
		\node [style=none] (10) at (1.25, 1) {};
		\node [style=none] (11) at (1.25, -1) {};
		\node [style=none] (16) at (-1.75, 1) {};
		\node [style=none] (17) at (-1.75, -1) {};
		\node [style=none] (18) at (-1.25, 1) {};
		\node [style=none] (19) at (-1.25, -1) {};
		\node [style=none] (20) at (-1.75, 1.75) {$\rA$};
		\node [style=none] (21) at (-1.75, -0.25) {$\rC$};
		\node [style=none] (22) at (1.75, 1.75) {$\rC$};
		\node [style=none] (23) at (4.25, 1.75) {$\rD$};
		\node [style=none] (24) at (4.25, -0.25) {$\rB$};
		\node [style=none] (25) at (1.75, -0.25) {$\rA$};
	\end{pgfonlayer}
	\begin{pgfonlayer}{edgelayer}
		\draw [in=-180, out=0, looseness=0.75] (1.center) to (3.center);
		\draw [in=-135, out=0] (0.center) to (4.center);
		\draw [in=180, out=45] (5.center) to (2.center);
		\draw [in=180, out=0] (10.center) to (8);
		\draw (8) to (7.center);
		\draw (11.center) to (9);
		\draw (9) to (6.center);
		\draw (16.center) to (18.center);
		\draw (17.center) to (19.center);
	\end{pgfonlayer}
\end{tikzpicture}